\documentclass[journal,comsoc]{IEEEtran}

\usepackage[nocompress]{cite}

\usepackage{graphicx}
\usepackage{color}
\usepackage{lscape}

\usepackage{amsmath}
\usepackage{amssymb}
\usepackage{amsfonts}
\usepackage{amsthm}
\usepackage{mathtools}
\DeclareMathOperator*{\argmin}{argmin}
\usepackage{booktabs}
\usepackage{multirow}

\usepackage{algorithm}
\usepackage{algpseudocode}

\usepackage{array}

\ifCLASSOPTIONcompsoc
  \usepackage[caption=false,font=normalsize,labelfont=sf,textfont=sf]{subfig}
\else
  \usepackage[caption=false,font=footnotesize]{subfig}
\fi
\begin{document}
\algblock{ParFor}{EndParFor}
\algnewcommand\algorithmicparfor{\textbf{parfor}}
\algnewcommand\algorithmicpardo{\textbf{do}}
\algnewcommand\algorithmicendparfor{\textbf{end\ parfor}}
\algrenewtext{ParFor}[1]{\algorithmicparfor\ #1\ \algorithmicpardo}
\algrenewtext{EndParFor}{\algorithmicendparfor}
\newcommand{\lgg}{\log_2\left(\frac{1}{\epsilon}\right)}
\newcommand{\bs}[1]{\boldsymbol{#1}}
\newtheorem{obs}{Observation}
\newtheorem{theorem}{Theorem}
\newtheorem{lemma}{Lemma}
\newtheorem{corollary}{Corollary}
\newtheorem{assumption}{Assumption}
\newtheorem{remark}{Remark}
\newcommand{\eacc}{$\varepsilon$-accuracy }
%
\title{Optimising cost vs accuracy of decentralised analytics in fog computing environments}

\author{\IEEEauthorblockN{Lorenzo Valerio, Andrea Passarella, Marco Conti}\\
\IEEEauthorblockA{Institute for Informatics and Telematics\\
National Research Council\\
Pisa, Italy\\
Email: \{l.valerio,a.passarella,m.conti\}@iit.cnr.it}
\thanks{\copyright 2021 IEEE.  Personal use of this material is permitted.  Permission from IEEE must be obtained for all other uses, in any current or future media, including reprinting/republishing this material for advertising or promotional purposes, creating new collective works, for resale or redistribution to servers or lists, or reuse of any copyrighted component of this work in other works. Accepted for publication. DOI: 10.1109/TNSE.2021.3101986}
}


%


\maketitle

\begin{abstract}
%
The exponential growth of devices and data at the edges of the Internet is rising scalability and privacy concerns on approaches based exclusively on remote cloud platforms. Data gravity, a fundamental concept in Fog Computing, points towards decentralisation of computation for data analysis, as a viable alternative to address those concerns.
Decentralising AI tasks on several cooperative devices means identifying the optimal set of locations or Collection Points (CP for short) to use, in the continuum between full centralisation (i.e., all data on a single device) and full decentralisation (i.e., data on source locations).
We propose an analytical framework able to find the optimal operating point in this continuum,  linking the accuracy of the learning task with the corresponding \emph{network} and \emph{computational} cost for moving data and running the distributed training at the CPs.
We show through simulations that the model accurately predicts the optimal trade-off, quite often an \emph{intermediate} point between full centralisation and full decentralisation, showing also a significant cost saving w.r.t. both of them. Finally, the analytical model admits closed-form or numeric solutions, making it not only a performance evaluation instrument but also a design tool to configure a given distributed learning task optimally before its deployment.

\end{abstract}


%
\section{Introduction} \label{sec:intro}
We are facing exponential growth in the number of personal mobile and IoT devices at the edge of the Internet. According to Cisco~\cite{Cisco2020} the number of connected devices will exceed by three times the global population by 2023. Most importantly, this is a trend that, apparently, not only will not slow down in the near future, but also will be the cause of the so-called ``data tsunami'', i.e., the explosion of the amount of data generated by these devices at the edge of the Internet. Ericsson~\cite{Ericsson:2020aa} foresees that global mobile data traffic will grow by almost a factor of 5 to reach 164EB per month in 2025.

Most of the value of data consists in the possibility of being processed and analysed to extract useful knowledge out of them. For example, this applies to many relevant cases of Industry 4.0 and Smart Cities where, in many applications, raw data are of little use, while the real value comes from the knowledge extracted through AI and Big Data analytics. The current approach for extracting knowledge from data is to centralise them to remote data centres, as many IoT architectures demonstrate \cite{Borgia:2014ab}. This is the case, among others, of the ETSI M2M architecture \cite{ETSI-arch}  where data are transferred from the physical locations where they are generated to some global cloud platform to be processed. 

Such an approach might not be sustainable in the long run because, despite the evolution of mobile networks, their capacity is growing only linearly \cite{Cisco2020}, which makes it impractical or simply impossible to transfer all the data to a remote cloud platform at reasonable costs. Furthermore, data might also have privacy and confidentiality constraints, which might make it impossible to transfer them to third parties such as global cloud platform operators. 
There are several scenarios where these constraints are relevant. One of the most important are applications in the Industry 4.0 (or Industrial Internet) area, where data analytics is one of the cornerstones \cite{Kagermann:2013aa}. 
However, companies might have severe concerns in moving their data to some external cloud provider infrastructure due to confidentiality reasons. On the other hand, they might not have the competences and resources to build and manage a private cloud platform. Moreover, real-time delay constraints might require that data elaboration or storage is performed at the edge, i.e., close to where data is needed, rather than in remote data centres.  
These trends and needs push towards a decentralisation of data analytics approaches towards the edge of the network, where the paradigms of edge computing, such as Fog Computing~\cite{Bonomi:2012aa},  Multi-access Edge Computing~\cite{Lopez:2015aa}, Cloudlets~\cite{Satyanarayanan:2009aa}, can address the aforementioned problems.\footnote{While there are architectural differences between these paradigms and in particularly between edge and fog computing, in the following we use the terms interchangeably, as those differences do not impact on the focus of the paper.}

Luckily, many machine learning (ML) algorithms at the basis of fundamental data analytics tasks such as classification, regression, and clustering admit distributed formulations, which can be used to implement decentralised data analytics. The main idea of distributed ML algorithms is to derive local models from partial datasets at several locations, and then exchange information across locations to refine the partial models.\footnote{Clearly, distributed ML assumes a certain degree of collaboration between nodes. This might be guaranteed either because nodes may be under the same controlling entity (e.g., the nodes inside a factory), or more in general through appropriate incentive schemes (e.g., for crowd-sensing applications \cite{Valerio:2016ab,Valerio:2016aa,Valerio:2017ab}). This is an orthogonal problem concerning the focus of this paper.}

In this paper, we tackle the issue of optimal configuration of such distributed ML algorithms. 
Specifically, given a set of nodes generating data, it is possible to identify a whole range of configurations, from fully centralised to fully decentralised, for a given ML algorithm. Each configuration is characterised by the set of \emph{collection points} where data are collected (possibly from more nodes generating them), and partial models computed (and refined via collaboration). As we will discuss in detail in the paper, given a target accuracy for the learning task, each configuration is characterised by (i) a network cost, required to move data to the collection points and exchange information for collaborative learning, and (ii) a computational cost at the collection points. The key contribution of this paper is to provide an analytical model that identifies the optimal operating point, i.e., the optimal set of collection points to be used.

More in detail, the contribution of the paper is as follows.  As explained in Section \ref{sec:learning_algo}, it is known that, by exchanging partial models updates between collection points and through multiple rounds of training on the local (partial) datasets, decentralised ML algorithms can achieve any target accuracy. However, this comes at a cost in terms of network traffic and computation. We define appropriate cost functions for both classes of costs in Section \ref{sec:learning_algo}. Thus, we develop an analytical model that, for any given number of collection points, provides the total cost of achieving a target accuracy. The model takes into consideration (i) the number of collection points, (ii) the amount of data at each collection point, (iii) the cost for transferring data on the network (iv) the cost for processing data. 
Based on the model, we are finally able to obtain the optimal operating point for a target accuracy, i.e., the optimal number of collection points such that the target accuracy is achieved at the minimal total cost.

We extensively validate the model through simulations. Specifically, we compare the predictions of our model with the optimal operational point (obtained from exhaustive search) for different types of computational cost functions. We show that the model is in general very accurate in predicting cost of training a decentralised ML algorithm, and identifying the optimal operating point. The model shows that the optimal operating point is almost always at an \emph{intermediate} aggregation level, between full centralisation and full decentralisation. Moreover, it also shows that a significant additional cost is paid when the ML algorithm works in the simplest fully-decentralised and fully-centralised configurations. This justifies the use of analytical tools like the one proposed in the paper to optimally configure decentralised ML operations. In fact, our model admits numerical and, in some cases, closed-form expressions for the optimal operating point, which are quite efficient to compute. Therefore, it can be used as a design tool to configure a system based on decentralised ML.

The key take-home messages of this paper are:
\begin{itemize}
\item \textbf{``decentralisation helps''}, since the optimal operating point is, in almost all cases, an intermediate point between full centralisation and decentralisation;
\item \textbf{``modelling tools help"}, since operating a decentralised ML at optimal operating point reduces quite significantly the cost with respect to more naive solutions;
\item the \textbf{shape of the compute cost function} plays a significant role in determining the optimal operating point. This means it can be tuned by edge service providers to drive the usage of their infrastructure according to specific policies;
\item the relative \textbf{costs of communication vs. computing} may push the optimal operating point either towards more or less decentralisation.
\end{itemize}

The rest of the paper is organised as follows. In Section~\ref{sec:related} we review the related literature. In Section~\ref{sec:problem}, we present the reference scenario, and we define the general structure of the cost model. In Section~\ref{sec:learning_algo} we present the distributed learning algorithm considered in the paper, while in Section~\ref{sec:costanalysis} we provide the mathematical formulation of its cost, for a given accuracy. In Section \ref{sec:settings} we introduce the datasets and the methodology used to evaluate our proposal. In Section~\ref{sec:perf}, we validate the model through simulations, and present a comparative analysis of the performance at the optimal operating point, as well as a sensitivity analysis of the optimal operating point against key parameters. Finally, Section~\ref{sec:conclusions} concludes the paper.

\section{Related work}
\label{sec:related}

The execution of  machine learning tasks at the edge, in the literature, is considered from several perspectives.

A very important body of work deals with distributed learning algorithms that are suitable for being executed in fog  scenarios. This is the case of the \emph{Federated Learning} Framework initially proposed by \cite{Konecny2015,McMahan2016c}. According to this framework, several devices coordinated by a central entity (i.e., a parameter server) and holding some local data, collaboratively train a global model (e.g., an artificial neural network) on the local data. During the process, the information exchanged between the device are only models' updates and other related information, but never raw data. Federated Learning is an iterative procedure spanning over several communication rounds until convergence is reached. Based on this paradigm, several modifications have been proposed concerning (i) new distributed optimisation algorithms~\cite{Wang2018,Amiri2019, Karimireddy2019,Mohri2019} and (ii) privacy-preserving methods for federated learning~\cite{Mao2018,MOTHUKURI2021619}.

Alternatively,  other approaches do not rely on a centralised coordinating server. In~\cite{Valerio:2016aa,Valerio:2017ab}  authors propose a distributed and decentralised learning approach based on Hypothesis Transfer Learning (HTL). Similarly to the Federated Learning framework, authors assume that several devices hold a portion of a dataset to be analysed by some distributed machine learning algorithm. The aim of \cite{Valerio:2016aa,Valerio:2017ab} is to provide a learning procedure able to train, in a decentralised way, an accurate model while drastically limiting the network traffic generated by the learning process. 
Other decentralised learning mechanisms propose methods based on \emph{gossiping} where the devices exchange the models updates only with their neighbourhood. This body of work builds on the initial ideas on asynchronous decentralised optimisation for randomised graphs \cite{Boyd:2006aa,Ram:2009aa} and time varying graphs \cite{Nedic2009}. Currently, the main challenges addressed in this part of the literature regard: scalability issues~\cite{Daily2018}, network reliability~\cite{Yu2019}, the reduction of the high communication costs induced by decentralised optimisation algorithms~\cite{Lian:2017aa,koloskova2019decentralized}, the need to cope with time-varying graphs~\cite{Nedic:2014aa,Nedic2015,Nedic:2017aa} and the problem of data heterogeneity among devices~\cite{Karimireddy2019}. 
Another body of work focuses on the coordination/orchestration of the learning process at the edge. Specifically, authors of \cite{Dey2018,Li2018} address the problem of offloading the computation needed to train a learning model from mobile devices to some edge/fog/cloud server. Authors of \cite{Li2016a} presents an edge solution in which the edge server coordinates the collaboration of several mobile devices that have to train models for object recognition. 

The closest work to the one proposed in this paper is presented in \cite{Valerio:2016ab} where authors studied if there exists a trade-off between the accuracy and the network performance connected to the execution of a decentralised analytics task in IoT scenarios. Although the methodology followed in \cite{Valerio:2016ab} is mostly empirical, the paper shows that there could exist such type of trade-off and that, once found, it is possible to optimise the configuration of the distributed learning process to reduce the related network cost significantly.
Another interesting solution with some similarity with this paper (i.e., the idea of using the convergence bounds of a learning algorithm) is presented in \cite{Wang2019}. Authors propose an adaptive federated learning mechanism suitable for resource-constrained edge systems. The solution exploits the convergence bounds of the learning algorithm to design a control algorithm that adapts the number of local model's updates during training and the number of communication rounds in order to be efficient. A similar solution is presented in \cite{Tran2019} where authors propose a control algorithm for optimising the trade-off between execution time and energy spent during the execution of a federated learning algorithm. 

In this paper, which extends our prior work in~\cite{Valerio:2017aa}, we neither propose a decentralised learning procedure for fog scenarios, nor a control algorithm for optimising the execution of such a procedure. Conversely, we address the problem of identifying, through an analytical tool, the optimal configuration of a broad class of decentralised learning algorithms, trough which a number of fog devices can (i) aggregate a certain amount of data and (ii) run a specific decentralised learning algorithm on them. Differently from \cite{Valerio:2017aa}, in our model we take into account, at the same time, both the network traffic generated by the data collection, the communications triggered by the distributed learning algorithm and the computational cost associated to the process. We solve the model both numerically and analytically (when possible), showing that it also admits closed-form expressions of the optimal aggregation level in certain cases.
We then analyse the accuracy of the model predictions compared to the optimal points found via simulations with exhaustive search over all possible configurations. We compare the cost penalty paid when configuring the ML task, instead of at the optimal point identified by our model, at the two simple extreme configurations, fully centralised and fully decentralised. Finally, we study, via the model, the sensitiveness of the optimal operating point with respect to the values of its key parameters.  To the best of our knowledge, this is the first time that such an analytical framework is proposed in the literature. A key distinguishing feature of our work is that our methodology can be applied to any distributed learning algorithm for which it is possible to find an analytical expression for the communication and the computation needed to converge to a fixed accuracy.
\section{Problem definition}
\label{sec:problem}

We consider a scenario (represented in Fig.~\ref{fig:system}) in which there are $m_0$ devices (e.g., high-end sensors of workers' devices in a smart factory, but also more powerful devices such as edge gateways) collecting data in their local storage. In our analysis, we assume that each device collects, on average, a certain amount of data $n_0$, each composed by \emph{d} features of equal size.\footnote{This feature model is grounded in the fact that, for example, smart devices are equipped with a number of sensors, whose readings can be considered together to enrich the extractable knowledge.} Therefore, the average total amount of data collected by all the devices is $N=m_0n_0d$.
Although this is a simplifying assumption taken for pure modelling reasons, it does not affect the generality of our results, as shown by the validation results presented in Section \ref{sec:perf}.

\begin{table}[ht!]
	\centering
	\caption{Main notation used in the paper.}
\begin{tabular}{ll}
\toprule
Symbol & Description\\ 
\midrule
$C_A$ & Algorithm Network Cost \\ 
$C_D$ & Data transfer Network Cost \\ 
$C_N$ & Network Cost \\ 
$C_P$ & Computational Cost \\ 
$m_0$ & num. of devices with data\\
$m_1$ & num. of data collection points\\
$n_0$ & num. of initial data per device \\
$N$ & total amount of data in the system \\
$R$ & number of round of the algorithm\\
$d$ & number of features of a data point \\
$\gamma$ & grouping parameter \\
$\theta$ & unitary cost for network communication \\
$\beta$ & unitary cost for computation \\
$\mu$ & parameter relating $\beta$ and $\theta$ \\
$\tau$ & FLOPS to compute a gradient \\
$\varepsilon$ & model accuracy w.r.t the optimal solution \\
$\kappa$ & condition number \\
$w$ & learning model parameters \\
$\omega$ & cardinality of $w$ \\
\bottomrule
\end{tabular} 
\end{table}

Moreover, we assume that all the devices in the system involved in the learning process are homogeneous and computationally capable of processing the amount of data they hold. We expect that even devices such as Raspberry PIs would belong to this class. Low-end devices (such as low-end sensors) involved in the data collection process are assumed to relay their data to more capable devices, where the distributed learning process can be performed. We point out that assuming homogeneous devices holds in all those contexts where devices are owned by a central authority that can decide the kind of hardware to install on them and it can control the process. This is the case, for example, of a fleet of (autonomous) vehicles collecting and processing data under a limited resource-budge; workers' smart devices or intelligent production machines in a smart factory produced by the same vendor that, instead of sending the collected data to the cloud it process them directly at the edge; smartphones provided by employers to employees that can extract knowledge from their data using schemes that fall within the framework of federated learning.

\begin{figure}[ht]
\centering
\includegraphics[width=\columnwidth]{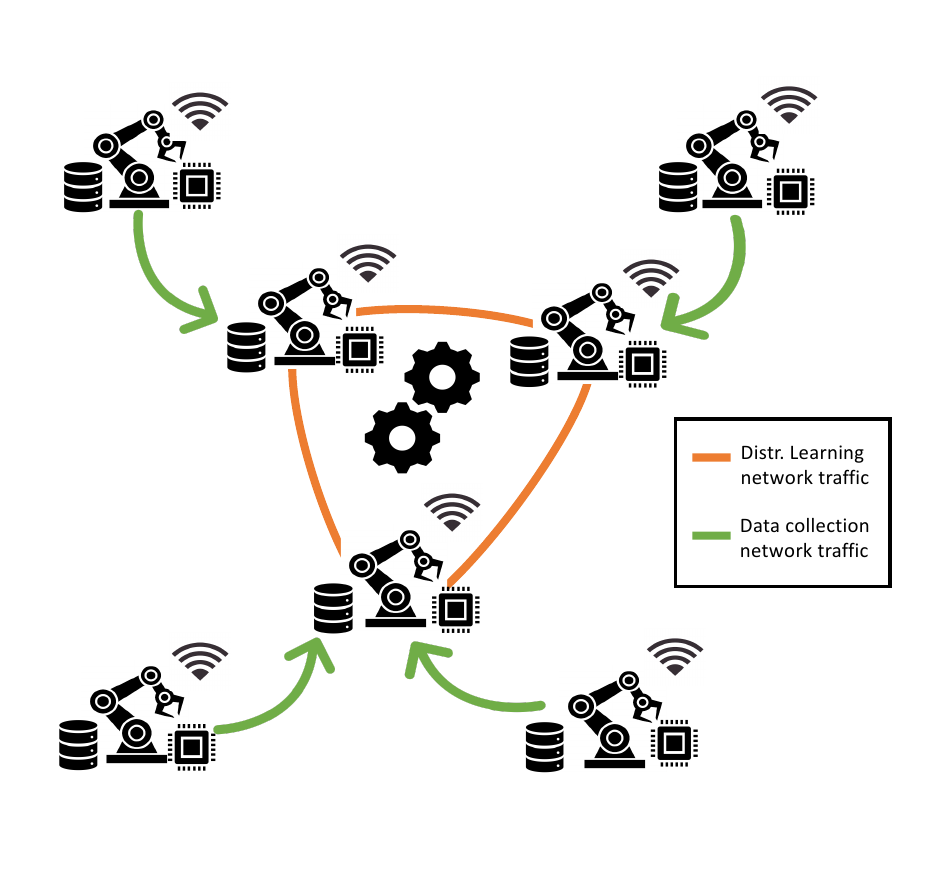}
\caption{Reference scenario.}
\label{fig:system}
\end{figure}

In our model, the data analytics process can be executed on a number $m_1$ of collection points (CPs) in the range $[1,m_0]$, where $m_1=1$ represents the fully centralised case and $m_1=m_0$ the fully decentralised one. The average number $n_1$ of data points collected on the $m_1$ CPs is  calculated as $n_1=\frac{n_0m_0}{m_1}$.

Our model provides, for a given learning task and a fixed target $\varepsilon$-accuracy (i.e., the optimality gap between the estimated and the optimal solution of the associated optimization problem), the optimal number of collection points, i.e., the number of collection points that minimises the associated cost.
 
The starting point for our model is the definition of the costs for the distributed training of the ML model. Specifically, we define the cost as
\begin{eqnarray}
    C &=& C_{N} + C_{P}
\end{eqnarray}
where $C_{N}$ represents the network traffic generated during the data collection/aggregation and the training process and $C_{P}$ is the term that represents the computational burden that CPs have to face in order to execute the distributed learning process (we discuss later on how to make the two cost terms homogeneous, and how to flexibly account for the relative importance of networking and computation costs in the total balance). 
Therefore, the cost function can be interpreted in many ways depending on the context. For example, it could be seen as the \emph{energy} cost imposed on the devices, but also as the \emph{monetary} cost imposed by network and computing operators to a third party requiring the data analytics service.

Without loss of generality, to have a single control parameter in the model, in the following we define $\gamma$ as $$\gamma=\frac{m_0}{m_1}$$ 
representing the level at which we group data at collection points (we call it \emph{grouping parameter}). Clearly, it holds that $\gamma \in [1,m_0]$, ranging from completely decentralised ($\gamma=1$) to completely centralised ($\gamma=m_0$).

The goal of this paper is to find the optimal value of $\gamma$ for a given target $\varepsilon$-accuracy. Our idea is that the quality of the final learnt model, which might be expressed by the model's accuracy, might connected to the optimality gap. Therefore, by binding the optimization variable $\gamma$ of our model to the a target optimality gap $\varepsilon$ we are implicitly creating a connection between the number of collection points and the final quality of the model.  Mathematically, we want to solve the following problem:
\begin{eqnarray}
    \argmin_{\gamma}  & C_N(\gamma) + C_P(\gamma)\label{eq:opt_problem}\\
\mathrm{s.t.} & \varepsilon_{\gamma} \leq \varepsilon \nonumber \\
& 1\leq \gamma \leq m_0 \nonumber
\end{eqnarray}
Where $\varepsilon_{\gamma}$ is the estimated $\varepsilon$-accuracy obtained by the distributed optimization algorithm for a given $\gamma$. The formal definition of $\varepsilon$ is provided in Section \ref{sec:costanalysis}, while in Section \ref{sub:methodology} we define how we estimate $\varepsilon_{\gamma}$.
The key challenge to solve Problem~(\ref{eq:opt_problem}) is to obtain an analytical formula for the two components of the cost, for a general enough class of learning algorithms. Before  presenting all the mathematical details of the solution of our model, in the next section we describe the class of learning algorithms that we consider in the rest of the paper. 

Note that the costs in Problem~(\ref{eq:opt_problem}) are derived under simplifying assumptions, to make their analytic forms tractable. The main assumptions considered are: (i) collection points are homogeneous; (ii) data available at collection points are iid; (iii) costs are derived as in the average case. We are aware some of these assumptions may not hold in our target environment. Therefore, in Section~\ref{sec:perf}, we assess through simulations to what extent violating these assumptions could affect the predictions of our model, which, we anticipate, proves to be quite robust.


\section{Distributed learning algorithms}
\label{sec:learning_algo}
In this section we briefly discuss the class of distributed stochastic optimisation algorithms that fall within the scope of our framework. We point out that the analysis in this paper can be applied to any ML algorithm that (i) admits a decentralised or distributed definition; (ii) analytic (even approximated) bounds on the number of communication rounds to achieve the sought accuracy and (iii) computation cost expressed as the complexity of the algorithm. The specific algorithm we consider in the following only serves the illustrative purpose to explain the process and derive concrete results in one relevant - though already quite significant - case.

Precisely, we focus on variants of the Stochastic Gradient Descent (SGD) that are i) communication and computation efficient by design and ii) equipped with the analytic expression of their communication and computation complexity bounds. A class of algorithms that owns these properties belongs to the category of the \emph{stochastic variance reduced methods}. The solutions in this category address one of the SGD's major deficiency: the high variance of the gradients. Briefly, SGD is one of the most widely used optimisation algorithms in both centralised and distributed settings. SGD was initially proposed to improve the efficiency of the Gradient Descent algorithm by substituting the evaluation of the full gradient (i.e., computed using all the data available) with an approximation computed on a smaller sample of data (i.e., from one sample to a subset of the whole dataset, called mini-batch). However, since the gradients computed by SGD are an approximation of the full gradient, they can have high variance, which is harmful for the convergence of the algorithm to a good solution. In SGD, the only way to contrast the variance issue is to vary or adapt the step-size of the updates, iteration by iteration, i.e., typically by setting a hyper-parameter called learning rate that decreases at each iteration of the algorithm. However, SGD with decreasing learning rate fails to achieve linear convergence.\footnote{The linearity refers to the number of updates w.r.t. the size of the dataset.} Stochastic reduced variance methods address the variance problem such that SGD achieves linear convergence, typically when the loss function is strongly convex. Some of the most prominent algorithms that belong this category, e.g., SAG\cite{Schmidt2017}, SAGA\cite{Defazio2014}, SDCA\cite{Schalev2013} and SVRG~\cite{Johnson:2013aa}, to mention a few, are designed for centralised settings but they have been adapted to work also in distributed ones.

In this paper, we focus on DSVRG~\cite{Lee:2015aa,Lee:2017aa}, one of first distributed versions of SVRG. We select it because it is well studied in the literature both theoretically and practically. Similarly to the other distributed learning algorithm of the same class, it is possible to derive, under few theoretical assumptions, the analytical expression for (i) the cost in terms of generated network traffic and (ii) the cost in terms of computation. In the following we provide a brief description of DSVRG and we comment the complexity bounds assumptions. Our aim is to provide the overall idea about how DSVRG works.\footnote{Interested readers are referred to \cite{Friedman:2001aa} for a detailed presentation.}

 In supervised learning, the goal is to learn the parameters of a mathematical model describing the relationship between data \emph{patterns} and the corresponding \emph{labels}. Formally, let us suppose that there exists a set of i.i.d. data points $\{s_1, s_2, \dots, s_N\}$ belonging to an unknown distribution $\mathcal{D}$. Each data point $s$ is a pair $(x,y)$ where $x\in\mathbb{R}^d$ is a vector of features (a pattern) and $y\in\mathbb{R}$ is the label associated to $x$. $y$ can be either a continuous value (for regression problems) or a discrete value (for classification problems). Moreover, let us suppose that the points of the dataset $D$ are grouped in $K$  separate (and possibly non-overlapping) subsets $ S_k$ physically stored on $K$ different devices. Therefore, the complete dataset is defined as $D=\bigcup_{k=1}^K S_k$. In distributed supervised leaning each device can only operate on local data. During training, it computes a loss function defined as follows.
\begin{equation}
\mathcal{L}_k(S_k,w)=\frac{1}{|S_k|}\sum_{i=1}^{|S_k|}\ell(q(x_{i,k},w),y_{i,k}) + \lambda\mathcal{R}(w)
\label{eq:loss_general}
\end{equation}
In Equation (\ref{eq:loss_general}), the first term is the average approximation error on the training set, while the second term is a regularisation factor that we describe in the following. The function $\ell(q(x_{i,k},w),y_{i,k})$ is a loss function computed on the device $k$, i.e. it measures the error of the model $q(x_{i,k},w)$ predicting the real values of $y_{i,k}$ for the data points in $S_k$. $w$ is the parameters' vector  of the model $q$. The function $\mathcal{R}(w)$ is a regularisation term whose purpose is to ease the search for the solution of the optimisation problem, producing a model which is less complex and, thus, less prone to over-fitting. Finally, the scalar $\lambda$ is a tuning parameter used to regulate the balance between the error term and the regularisation term. 

In the following we denote with $w$ the global model's parameters to be learned (note that, the model parameters will be also referred to as weights, without loss of generality). The goal of the learning algorithm (e.g., DSVRG) is to find the values of parameters $w$ that minimise the average loss over all devices. Formally, the objective is to find the parameters $w^*$ such that
\begin{equation}
    \label{eq:ermr}
    w^*=\argmin_{w}  \frac{1}{K}\sum_{k=1}^K \mathcal{L}_k(S_k,w)
\end{equation}

To solve the problem (\ref{eq:ermr}) DSVRG assumes the presence of a central unit whose task is to coordinate the communication between the devices and monitor the overall learning process. 
A typical assumption is that the central unit (also called centre) has complete knowledge of the system, i.e. information about devices, about the data they hold, etc. We will discuss in the following how to relax this assumption, thus making the algorithm decentralised.
%

The algorithmic description of DSVRG can be split into two logically separate procedures, one executed by the ``centre" and one executed by each of the other devices. For ease of explanation, in the following we describe DSVRG as if a centre would actually exist. However, in our scenario we consider that collection points assume the role of centre in a round robin fashion, thus making DSVRG totally decentralised and, as we show, also reducing the overall network cost. 

At the beginning ($r=0$), the centre broadcasts an initial estimate of the solution $\tilde{w}_{r}$ to all the devices. Each of them computes in parallel the new gradient over all data points contained in their local datasets and sends it back to the centre. The centre computes the average gradient $\tilde{h}_{r}$ from those collected by the client devices. At this point, the centre selects a single device $k$ and sends to it the average gradient $\tilde{h}_r$. Device $k$ updates for $T$ local iterations its local estimate of the weights $w_{r,t}^{(k)}$ using Eq.~(\ref{eq:wupdate}) and the current  estimate of the global model as in Equation (\ref{eq:wgiter}), whose meaning we explain next. Then, it computes the new global model $\tilde{w}_{r+1}$ (using Eq.~\ref{eq:wgupdate}) and sends it to the centre. Finally the centre is rotated, i.e., the centre changes in a round robin fashion.

\begin{eqnarray}
w_{r,t+1}^{(k)} &=& w_{r,t}^{(k)}-\eta\left(g(w_{r,t}^{(k)},s_{t}^{(k)}) - g(\tilde{w_r},s_{t}^{(k)}) + \tilde{h}_r\right)\label{eq:wupdate} \\ 
\bar{w}_{r,t+1} &=&\frac{w_{r,t+1}^{(k)}+t \bar{w}_{t}^{(k)}}{t+1} \label{eq:wgiter}\\
\tilde{w}_{r+1} &=& \bar{w}_{r,T-1}^{(k)} \label{eq:wgupdate}  
\end{eqnarray}


%
where $t=0,\dots,T-1$, $w_{r,0}^{(k)}=\tilde{w}_r$ and $\bar{w}_{r,0}^{(k)}=\bs{0}$. The rationale of Equation~(\ref{eq:wupdate}) is to update the estimate of the parameters with a quantity that is  proportional to (i) the local gradient ($g(w_{r,t}^{(k)},s_{t}^{(k)})$), (ii) the gradient of the global current solution ($g(\tilde{w}_r,s_{t}^{(k)})$) and (iii) the average global gradient ($\tilde{h}_r$). Specifically, the local gradient in (i) is computed by using only the current (at time step $t$) values of the local weights $w_{r,t}$ on the local dataset. The gradient in (ii) is computed using the current (at time step $t$) estimate of the overall parameters ($\tilde{w}_r$) on the local dataset. Each individual collection point $k$ has all the input parameters to compute both gradients. Therefore, the update of the local weights through Equation~(\ref{eq:wupdate}) takes into account both local and global models parameters.
 The rationale of Equation~(\ref{eq:wgiter}) is to update the global weights as a weighted average of the old value in the previous round and the updated local weights computed through Equation~(\ref{eq:wupdate}). Finally, the last update (at $T-1$) of the average weights $\bar{w}_{r,T-1}^{(k)}$ becomes the updated global model for the next round $r+1$ as in Equation~(\ref{eq:wgupdate}).
%
This procedure is repeated for a number $R$ of rounds.

Note that the algorithmic description presented here is only a short summary of the complete ones. Their purpose is to provide to the reader an intuitive idea of how DSVRG works. For more details on the internals of the algorithm, the interested reader can refer to the original paper \cite{Lee:2015aa,Lee:2017aa}. 

As anticipated, the asymptotic convergence behaviour of DSVRG can be summarised with the following communication bound $$O\left((1+\frac{\kappa}{n_1})\lgg\right)$$, whose precise meaning will be discussed in Section~\ref{sec:costanalysis}. The assumption made to derive this bound are: (i)~the individual loss function is convex and L-smooth,\footnote{This means that the individual loss function is differentiable and its gradient is L-Lipschits continous.} (ii)~the average loss function is strongly convex, (iii)~the data are i.i.d. among devices and (iv)~each device has access to extra data during training, i.e., in the original paper, in order to guarantee good working conditions for the algorithm the authors assume that devices can redistribute part of their data with the other devices. The extra data is used for computing the update on the $k$-th device, while the local data is used for computing the global gradients. Since in this paper we target fog/edge scenario where data might be non-iid and devices might not be allowed to access the extra data, in the experimental evaluation we violate willingly these assumptions, to test the robustness of our model in conditions that are not compliant to the theoretical assumptions of the learning algorithm considered. Therefore, the only assumption satisfied regards the properties of the loss function, as it will be clear in the following. 

\section{Cost Model}
\label{sec:costanalysis}
\subsection{Network Cost of DSVRG}

To analyse the network cost of DSVRG it is necessary to first define the $\varepsilon$-accuracy of the algorithm. Remember that the goal of our analyses is to identify the optimal operating point given a target \eacc.

The \eacc is defined as $$|\mathcal{L}(\tilde{w})- \mathcal{L}(w^*)|\leq \varepsilon .$$ It means that for the model with parameters $\tilde{w}$ the value of the objective function is far from the optimal solution $w^*$ at most by $\varepsilon$. The optimal solution $w^*$ is in general unknown. 
We recall that the value of $\varepsilon$ must not be confused with the value of the generalisation performance of the learning algorithm (i.e., model's accuracy). Clearly, these quantities are related to each other because a more accurate solution translates into a better quality of the model. However, quantifying in advance the exact relation between the two is very difficult (if not impossible) and strongly problem-dependent.

%
%

As we will show in the following, $\varepsilon$ directly affects the number of communication rounds and therefore both the communication and compute costs.
Precisely, as stated in \cite{Lee:2015aa}, the relationship between the $\varepsilon$-accuracy and the number of communication rounds between the centre and the devices is expressed by

\begin{equation}
    \label{eq_rounds}
    R = \left(1+\frac{\kappa}{n_1}\right)\log_2\left(\frac{1}{\varepsilon}\right)
\end{equation}
where $n_1$ is the size of the local dataset held by each Collection Point and $\kappa$ is the condition number of the problem. 
The condition number is defined as the ratio between the first and the last eigenvalues of Hessian of the empirical loss function.
Precisely, $\kappa$ measures how much the output value of the function can change for a small change in the input argument. Computing the exact value of $\kappa$ is computationally expensive, but as specified in \cite{Lee:2017aa} we can consider it proportional to $\sqrt{N}$, where $N$ is the cardinality of the entire dataset. Note that the number of rounds depends on the aggregation level, as $n_1$ (the number of data points available at a collection point) is equal to $\gamma n_0$.

In DSRVG, the amount of communication generated in a single round can be computed as in Equation (\ref{eq:cr}). Note that we consider as unitary cost of communication the cost of sending a single model (i.e., all its parameters).
\begin{equation}
\label{eq:cr}
 C_R = 2(m_1-1)\omega
\end{equation}
Specifically, for each round the centre sends to all the other $(m_1-1)$ Collection Points a parameters vector $w$ and receives from them $(m_1-1)$ gradient vectors. The size of the parameters vector and of each gradient is expressed by $\omega$.
For the sake of clarity, in the following we describe how we derive the communication costs. Without loss of generality, we assume that the node acting as centre also executes the local updates, i.e,. Eq.(\ref{eq:wupdate}-\ref{eq:wgupdate}). This means that the centre does not generate traffic for sending the vector $\tilde w_t$ and the average gradient $\tilde h_t$ to itself. Finally, we assume that, at a given round, the current centre indicates the centre for the next round while sending the new model, thus with negligible network cost. Again without loss generality, we consider that the centres are selected according to a round robin policy. The specific policy does not impact on the network cost as long as the current centre has locally the information required to identify the next one.

The total amount of traffic generated to obtain an $\varepsilon$-accurate solution is therefore
\begin{equation}
    C_A = C_R * R = 2\omega(m_1-1)(1+\frac{\kappa}{n_1})\log_2\frac{1}{\varepsilon}
\end{equation}
In addition to $C_A$, which is the amount of traffic generated by the learning algorithm, we also have to account for the network traffic to collect raw data at the $m_1$ collection points, defined as $C_D$. This is easily computed as each node that is not a collection point ($m_0-m_1$ nodes in total) sends the local data ($d n_0$ on average) to one collection point. The final cost in terms of network traffic is thus provided in Equation (\ref{eq:cost_model}), where $\theta$ is the cost for transmitting a single feature between two nodes.\footnote{More precisely, it is the cost of sending a message of average size, computed over the sizes of messages required to transfer raw data (each data point generates $d$ messages) and of the messages required by DSVRG (each model generates $\omega$ messages)} Note that $\theta$ highly depends on the specific communication technology used.


\begin{equation}
    \label{eq:cost_model}
    C_{N} = \theta(C_A + C_D) =  \theta C_A + \theta(m_0-m_1)d n_0
\end{equation}

Note that Equation~(\ref{eq:cost_model}) is a model of the average case, where in particular all collection points store the same amount of data, which are iid. In Section~\ref{sec:perf} we validate the accuracy of the model by relaxing these assumptions.

\subsection{Computational Cost of DSVRG}
First of all, clearly only collection points incur in computation costs. We model the computation required by all collection points with Equation~(\ref{eq:prepsi}) below:
\begin{equation}
P(\gamma)= G(\gamma)R(\gamma).
\label{eq:prepsi} 
\end{equation}
In Equation~(\ref{eq:prepsi}), $R(\gamma)$ denotes the number of rounds to achieve $\varepsilon$ accuracy, and is again provided by Equation~(\ref{eq_rounds}) (in this case we only make it explicit the dependence on $\gamma$). Moreover, $G(\gamma)$ is the number of floating point operations needed to compute the gradients by all collection points during a round. Specifically, this is equal to the number of data points over which the gradients need to be computed multiplied by the cost for computing a single gradient, that we refer to as $\tau$. Specifically, $\tau$ is the number of floating-point operations (FLOPS) needed to compute a single gradient vector, given an input pattern, and it is typically proportional to the size of the model, $\omega$ in our case. According to DSVRG the number of gradients to compute at each collection point is equal to the number of local data points, i.e., $n_1 m_1$. In addition, the centre needs to compute one additional gradient per local point, needed by the term $g(\tilde{w}_t,s_{k,t})$ in Equation~\ref{eq:wupdate}. Thus, the total number of gradients to compute is $n_1 (m_1+1)$.\footnote{In the paper we consider the single gradient case because it is how DSVRG is designed. Note that it is possible to adapt the term for the FLOPS-count to the case where an algorithm takes in input batches of samples and the compute unit is optimised to treat them efficiently. A way might be to apply a scale factor to $P(\gamma)$ to catch such an improvement. However this extension is not within the scope of this paper.}
Therefore, $G(\gamma)$ becomes as follows:
\begin{equation}
    G(\gamma) = \tau n_1 (m_1+1)
    \label{eq:G-gamma}
\end{equation}

From Equation~(\ref{eq:prepsi}) we obtain the final expression of the cost for computing 
the data $C_P(\gamma)$ at a certain aggregation level $\gamma$: 
\begin{equation}
 C_P(\gamma)=\psi(\gamma)P(\gamma).
\end{equation}
$\psi(\gamma)$ represents the cost of a single floating-point operation. We assume that $\psi(\gamma)$ depends on the total number of operations needed for the computations at each collection point, or, better, on the aggregation level $\gamma$. This allows us to model the cost of using computation on resource-limited collection points, as well as the economy of scale available at big collection points run by infrastructure operators, as well as their strategies to attract compute services. For example, in cases where collection points are resource-constrained devices (e.g., Raspberry PIs), computations might easily saturate local resources. Therefore, it is reasonable to use a cost function that scales super-linearly with the amount of data handled by each collection point, and thus with the aggregation level. On the other hand, when CPs are nodes of an operator's infrastructure, computation might be very cheap (assume, e.g., cases where CPs are virtual machines running on a set of edge gateways, typically over-provisioned with respect to the specific needs of the learning task). In such cases, it might be reasonable to model the cost as a linear or even sub-linear function of the amount of data each CP handles. 
To capture all these possibilities at once, we define $\psi(\gamma)$ as in Eq.~(\ref{eq:psidef}):
\begin{eqnarray}
\label{eq:psidef}
\psi(\gamma) &=&   \beta f(\gamma)
\end{eqnarray}
 where $\beta$ is the cost of a single floating point operation required to compute gradients and $f(\gamma)$ captures the sub-/super-linear dependence with the number of data handled by a CP. It is thus defined as in Equation~(\ref{eq:deffg}). 
\label{eq:deffg}
\begin{equation}
f(\gamma)  =  \gamma^{\alpha}
\end{equation}
where $\alpha$ is a tuning parameter that we use to control the shape of $f(\gamma)$. 
Precisely, by tuning appropriately $\alpha$ we can obtain three different regimes of cost, as shown in Table~\ref{tab:regimi}. In the paper, we analyse and solve our model with each one of them.
Note that this formulation of computation costs does not only capture technical constraints, but also simple service provider policies related to billing of computing infrastructures. We anticipate that the values used in the rest of the paper for $\alpha$ are only illustrative of the behaviours of possible configurations in the three mentioned regimes.

\begin{table}[ht!]
    \centering
    \caption{Possible parametrisation of $f(\gamma)$.}
    \begin{tabular}{ll}
        \toprule
        $f(\gamma)$ & $\alpha$ \\ 
        \midrule
        linear & $1$  \\ 
        sublinear & $(0,1)$ \\ 
        superlinear & $(1,\infty)$ \\ 
        \bottomrule
    \end{tabular} 
\label{tab:regimi}
\end{table}

Before presenting, in the next section, the analysis of the optimal operating point that can be obtained through our model, let us highlight a few important features. First, in our model, the balance between the communication and compute costs are defined by the two parameters $\theta$ and $\beta$. The former defines the cost of sending a single parameter (e.g., a float), while the latter defines the cost of performing a floating point operation in the gradient computation. Moreover, the dependency of the model with the specific algorithm considered is limited to the formula of the number of rounds $R(\gamma)$ to achieve $\varepsilon$-accuracy (Equation~\ref{eq_rounds}) and the number of floating point operations to compute all gradients, $G(\gamma)$ (Equation~\ref{eq:G-gamma}). The rest of the model does not depend on the specific definition of the DSRVG algorithm, and can thus be used also with other ML algorithms. Also remember that DRSVG already represents a broad class of ML algorithms itself.


\section{Cost model analysis}
\label{sec:model_analysis}
In this section we analyse the objective function of the optimisation problem defined in Section \ref{sec:problem} for the case of the learning algorithm presented in Section \ref{sec:learning_algo}.
For the sake of clarity, let us rewrite the cost model expanding $C_A(\gamma)$ and $C_P(\gamma)$:
\begin{eqnarray}
	C(\gamma)  &=& \theta\left((C_A(\gamma) + C_D(\gamma)\right)+C_P(\gamma) \nonumber \\
    &= &  2 \omega \theta \left(\frac{m_0}{\gamma}-1\right)R(\gamma) +\nonumber\\
    & & +   \theta\left(m_0 - \frac{m_0}{\gamma}\right)n_0 d +\nonumber\\
    & & +   \beta f(\gamma) \tau(n_0m_0+n_0\gamma)R(\gamma)
    \label{eq:cg}
\end{eqnarray}


%
First we enclose the conditions for which the objective function in Eq.(\ref{eq:cg}) is convex in the following theorem. 


\begin{theorem}[Convexity]
\label{th:convex}
There exist $\alpha_1,\alpha_2 \in (0,1)$ where $\alpha_1<\alpha_2$ such that $C_A(\gamma) + C_P(\gamma)$ is convex for $\alpha \in (0,\alpha_1] \cup [\alpha_2,\infty)$  and  $\gamma \in [1,\infty)$.
\end{theorem}
\begin{proof}
  See Appendix \ref{app:proofthcvx}.
\end{proof}
\begin{remark}
\label{rm:convex}
The $C_D(\gamma)$ is a concave function (as proved in Appendix \ref{app:proofthcvx}). Therefore the convexity of the Eq. (\ref{eq:cg}) depends on the magnitude of $C_D(\gamma)$ w.r.t. $C_A(\gamma)+C_P(\gamma)$. Considering Theorem \ref{th:convex}, when $C_D(\gamma)$ is negligible in comparison to $C_A(\gamma)+C_P(\gamma)$ then the function $C(\gamma)$ is convex. 
\end{remark}

Assuming the cost for transferring data to collections points negligible with respect to the sum of the other two terms, according to Theorem (\ref{th:convex}) and Remark (\ref{rm:convex}), when the cost for computing grows either linearly or super-linearly with the level of aggregation $\gamma$, we can identify a unique solution that minimises the objective function (\ref{eq:cg}). Conversely, for sub-linear aggregation cost, the convexity depends on the specific value taken by $\alpha$ in the range $(0,1)$. The analytical form of $\alpha_1,\alpha_2$ is provided in Appendix \ref{app:proofthcvx}.

We have been able to find the closed-form solution of the optimal value only for following two cases: (i) when we consider the network cost alone, i.e., $C(\gamma)=\theta(C_A(\gamma)+C_D(\gamma))$, as already presented in~\cite{Valerio:2017aa} and reported in Theorem~\ref{th:net} for completeness and, (ii) when $f(\gamma)$ is linear, i.e. with $\alpha=1$, as stated by Theorem~\ref{th:lin}. Conversely, for the general case, i.e, when $\alpha$ is left symbolic, we have not been able to find a closed-form solution but, as we will show in Section \ref{sub:validation}, it is possible to find the numerical solution of our model using a standard solver. 

The case considering network cost only is represented in our model by setting $\beta=0$. Therefore, Eq. (\ref{eq:cg}) becomes:
\begin{eqnarray}
C(\gamma) &=& C_N(\gamma) \nonumber\\
 & =& \theta(C_A(\gamma) + C_D(\gamma)) \nonumber \\
& = & 2 \omega \theta \left(\frac{m_0}{\gamma}-1\right)R(\gamma) +\nonumber\\
 &   & +   \theta\left(m_0 - \frac{m_0}{\gamma}\right)n_0 d 
 \label{eq:netonly}
\end{eqnarray}
The solution of Eq. (\ref{eq:netonly}) is provided by the following theorem.  


\begin{theorem}
	\label{th:net}
	When $\beta=0$, $C(\gamma)$ admits only one solution in $\mathbb{R}$. Its expression is given in Equation (\ref{eq:netopt}).
	\begin{equation}
	\label{eq:netopt}
	\tilde{\gamma} = \frac{4 \kappa \lgg m \omega}{d m n^{2} + 2 \kappa \lgg \omega - 2 \lgg m n \omega}
	\end{equation}
\end{theorem}
\begin{proof}
	The proof is provided in \cite{Valerio:2017aa} where it is shown that the sign of first derivative of $C_N(\gamma)$ before $\tilde{\gamma}$ is negative and after  $\tilde{\gamma}$ is positive, identifying $\tilde{\gamma}$ as a minimum. 
\end{proof}

%
Let us now present the analytical solution for the case in which $\beta>0$ and the computational cost is linear, in the following Theorem.  

\begin{theorem}
	\label{th:lin}
	Provided Theorem~\ref{th:convex}, when $f(\gamma)=\gamma$, $C(\gamma)$ admits only one minimum $\tilde\gamma$ in the range $[1,m_0]$ and it can be found solving Equation (\ref{eq:gopt}) in $\gamma$. The expression of $\tilde\gamma$ is provided in Appendix \ref{app:defs}.
	\begin{equation}
	\label{eq:gopt}
	C'(\gamma)=0
\end{equation}

\end{theorem}
\begin{proof}
	See Appendix \ref{app:profthlin}.	
\end{proof}
We can exploit the result of Theorems \ref{th:net} and \ref{th:lin} to identify the optimal operating point of our problem in the range $[1,m_0]$ as in Equation (\ref{eq:minnet}). Specifically, in the following $\widehat\gamma$ represents the estimated optimal grouping level obtained through our model, while $\tilde{\gamma}$ is the value of $\gamma$ that minimises function $C(\gamma)$. $\widehat\gamma$ and $\tilde{\gamma}$ may differ, if the latter falls outside of the admissible values of $\gamma$, i.e., $[1,m_0]$. Equation (\ref{eq:minnet}) exemplifies this behaviour.

\begin{equation}
\widehat{\gamma} =\left\{\begin{array}{ll}
1 &  \tilde{\gamma}< 1 \\
\tilde{\gamma}  & 1\leq\tilde{\gamma}\leq m_0 \\
m_0  &\tilde{\gamma}>m_0\end{array} \right.
\label{eq:minnet}
\end{equation}

The interpretation of Equation (\ref{eq:minnet}) is as follows and holds for the both Theorems. If $\tilde{\gamma}$ is below $1$, the function $C(\gamma)$ crossing the domain $[1,m_0]$ is increasing for $\gamma>\tilde\gamma$. Therefore the only viable option for minimising the  costs  is to adopt a fully distributed configuration, leaving data on the source locations. 
Conversely, if $\tilde{\gamma}$ is beyond $m_0$ then  $C(\gamma)$ crossing the domain $[1,m_0]$ is strictly decreasing when $\gamma<\tilde\gamma$. This means that in this case the minimum cost is equivalent to centralised all the data on a single device. 
For all the case in which $\tilde{\gamma}$ gets values in the range $[1,m_0]$ the best option is represented by an intermediate solution. In the simplest case when $\beta=0$, it is possible to provide very intuitive explanations of the form of $\tilde \gamma$, see~\cite{Valerio:2017aa}. 
\section{Evaluation settings and methodology}
\label{sec:settings}
 In this section we present the reference datasets  used in our evaluation (Section~\ref{sub:dataset}), the evaluation methodology and settings  (Section~\ref{sub:methodology}).  Without loss of generality, in the following we consider a binary classification task, and, as customary in the literature (e.g.,~\cite{Valerio:2017aa}), we use the logistic function as error function, i.e, the function $\ell$ of Eq.~(\ref{eq:loss_general}), and the squared norm as regularisation term (i.e., $\mathcal{R}(w)=\|w\|_2^2$).\footnote{The squared norm is a common regulariser in ML. Its physical meaning is to penalise the model's parameters with large values, thus preventing over-fitting issues.}. The loss function thus becomes:

\begin{multline}
\mathcal{L}_k(S_k,w){=}\frac{1}{|S_k|}\sum_{i=1}^{|S_k|}\log(1+\exp(-y_i(x_i^Tw))) + \lambda\|w\|^2_2
\end{multline}
where $\lambda=0.01$. This value has been selected after a parameter search and kept fixed during the whole evaluation. 
\subsection{Datasets description}
\label{sub:dataset}
We base our analysis on three real-world and publicly available reference datasets: Covtype,\footnote{Dataset available at https://archive.ics.uci.edu/ml/datasets/Covertype} Gisette,\footnote{Dataset available at https://archive.ics.uci.edu/ml/datasets/Gisette} and MNIST.\footnote{Dataset available at https://www.openml.org/d/554}
\begin{itemize}
\item Covtype (CT) contains real-world observations related to an environmental monitoring task. It contains 581012 observations, where each one is a vector of 54 features containing both cartographic and soil information, corresponding to a $30m\times30m$ area of forest. The learning task is to use this information to predict what is the main kind of trees covering the area. 
\item MNIST (MN) is a dataset containing images  representing handwritten digits, from 0 to 9. The problems is to classify the handwritten digit. The dataset contains 70k images of size  28x28 (i.e., 784 features). 
\item Gisette (GS) is a handwritten digit dataset based on MNIST. The problem is to separate the highly confusable digits '4' and '9'. This dataset is one of five datasets of the NIPS 2003 feature selection challenge. With respect to original MNIST, the data were modified to make the feature selection more difficult. In particular, each vector contains i) a subset of the initial features mapped into a higher dimensional space and ii) distractor features (i.e., noise) with no predictive power for the two classes. Moreover, the spatial relationships between features are removed through randomisation. The dataset contains 13500 points of 5000 features each.
\end{itemize}
These datasets have been selected because they are reference benchmarks in the ML literature and provide  sufficient data to test our model in various configurations, as clarified in the following. Although they are not necessarily representative of the scenarios described in Section~\ref{sec:intro} (such as Industry 4.0), the generality of our results is not affected because the main purpose of the paper is to assess the accuracy of our analysis and quantify the advantages of decentralised learning, which does not depend on the specific environment used for evaluation.
 Note that, Covetype and MNIST contain seven and ten classes, respectively. Since in this paper we test our model by solving a  binary classification problem, we build the binary datasets following the One-vs-Rest procedure, i.e., the positive examples of the binary dataset belong to one of the classes while the negative examples are randomly selected from the rest of the classes. In the following and for both MNIST and Covtype we present results only for one among all the possible binary dataset. Specifically, results refer to the dataset with positive class 7 for Covtype and 0 for MNIST. From experiments we realised that results from these dataset do not harm the generality of the analysis. Moreover, note that due to space limits we cannot include the results for all the combination of parameters and datasets considered during the performance evaluation. 
 
\subsection{Methodology}
\label{sub:methodology}
In order to validate and analyse our model, we use the following procedure. 
Each dataset is split into training and test sets, with a proportion $(80\%,20\%)$. 
In the following of the paper we consider data heterogeneity among devices both in terms of size and class representation (hereafter, both referred at once as \emph{heterogeneity}). This allows us to be compliant with the typical assumptions of federated learning, and to relax key simplifying assumptions under which our model has been derived, validating its robustness.
The procedure to simulate data collection in a way that preserves heterogeneity at both the initialisation of the systems and after the data aggregation on CPs is provided in Section \ref{sub:het}.

We set a number $m_0$ of devices, each one holding a certain number of observations of the training set. 
We select a target accuracy $\varepsilon$, and we impose, following~\cite{Lee:2015aa, Lee:2017aa}, that $\kappa$ is proportional to $\sqrt{N\cdot d}$, where $N=m_0\cdot n_0$ is the total amount of data held by the devices in the system. In order to make communication and computation costs comparable we assume that $\theta, \beta$, which in our model are respectively the generic cost for sending a feature and the generic cost of an operation for computing gradients, are expressed in energy units (Eu). We also assume that the two quantities have the following relationship: $\beta= \theta *\mu$, where $\mu\geq 0$.  In this way, with a single parameter $\mu$, we can control the contribution of the two terms of the cost function. Note that the optimal point identified by our model is insensitive to the specific values of $\beta$ and the $\theta$, but only depends on their ratio, i.e. $\mu$.

In the following analysis, we evaluate the optimal operating point, derived by solving our analytical model as explained in Section~\ref{sec:costanalysis}. We refer to this value as $\hat{\gamma}$. To validate our model, we proceed as follows. We run the distributed learning algorithm for all values of $\gamma$ in the range $[1,m_0]$. For each run, the algorithm stops when the estimate $\varepsilon_\gamma$ of the \eacc drops below the target value $\varepsilon$ (as in the definition of problem (\ref{eq:opt_problem})). In this paper the $\varepsilon_\gamma$ is estimated as the change of parameter values across iterations, measured via gradients. The intuition is that when the updates do not change significantly from one iteration to the next, the solution is close to the optimal one (assuming a convex loss function as we do in the paper). Specifically, since DSRVG is a gradient-based algorithm, the accuracy is estimated as the squared norm of the gradients, i.e.,
$$\varepsilon_{\gamma} = \|\nabla \mathcal{L}(\tilde{w},s)\|_2$$
where $\nabla\mathcal{L}(\bar{w},s)$ is the gradient of the model evaluated on point $s$.
 For each level of aggregation, we collect i) the number of communication rounds, ii) the number of FLOPS executed to reach the target $\varepsilon$-accuracy and iii) the amount of data collected on the CPs. 

All the simulations have been repeated ten times in order to reduce the variance of the results. Therefore, the results reported in this paper are average values for which we report also the confidence intervals at 95\% level of confidence.
These values are used as input to Equation~(\ref{eq:cost_model}) to calculate the empirical costs $C_A(\gamma),C_D(\gamma),C_P(\gamma)$ at the various aggregation levels. We define as \emph{empirical} optimal operating point the value of $\gamma$ for which we obtain minimum cost, hereafter referred to as $\gamma^*$. Notice that, for a given aggregation level $\gamma$, the analytical expression of the total cost and the one obtained via simulation may differ, as in the former case we are estimating both the number of rounds needed for convergence and the number of computations required. In turns, this means that the optimal operating point estimated by our analytical model is an approximation of the real optimal operating point found in simulations (which requires an exhaustive search over the $\gamma$ parameter range). Comparing the values of $\gamma^*$ and $\hat{\gamma}$, as well as the costs incurred at those aggregation levels, we are able to validate the accuracy of our solution. Specifically, we define \emph{overhead} (OH) the percentage additional cost incurred when using an aggregation level $\hat{\gamma}$ instead of $\gamma^*$.

\begin{table}[t]
\caption{Default model's parameters for all the considered datasets: Covtype (CT), Gisette (GS) and MNIST (MN) }
\begin{center}
\scriptsize
\begin{tabular}{lllllllll}
\toprule
  & N & $m_0$ & $n_0$ & d & $\varepsilon$ &  $\theta$ & $\mu$ & $\beta$\\
 \cmidrule{2-9}
 \addlinespace
CT & $\sim$32k & 400 & 112 & 54 &  &   & & \\
GS & 5.6k & 50  & 112 & 5000 & $10^{-3},\dots, 10^{-7}$ &  $1$ & $10^{-4}$ & $10^{-4}$\\
MN & $\sim$ 11k & 50 & 219 & 784 & &   &  & \\
\addlinespace
\bottomrule
\end{tabular} 
\end{center}
\label{tab:settings}
\end{table}

After validating the model, in Section~\ref{sub:comparison} we compare the cost achieved at the optimal operating point with respect to a fully centralised $\gamma=m_0$ and a fully decentralised $\gamma=1$ configurations. This allows us to highlight the advantage of configuring the distributed ML system through our model, as opposed to use two straightforward configurations. Finally, in Section~\ref{sub:sensitivity} we present a sensitivity analysis of the optimal configuration point and the related cost, by varying key parameters, as follows. Table~\ref{tab:settings} shows the default values for the main parameters used in the following analysis. The size of the datasets $N$ is derived from the classes that we selected. The number of devices $m_0$ generating data is used to represent a densely populated fog environment. The values of $\varepsilon$ are typical accuracy levels considered in the literature. As for the relative costs of computation vs. networking (i.e., $\mu$) we consider, by default, a case where computation is much cheaper ($10^{-4}$) than communications, under the assumption that a wireless transmission of data  is more power consuming than processing the same amount of data, e.g., computing a gradient. Note that in Section~\ref{sub:sensitivity} we analyse the impact of varying all of these parameters, considering also cases where computation is basically for free.

\subsection{Heterogeneous data collection and aggregation}
\label{sub:het}
To simulate a realistic, heterogeneous data collection process, each device draws, from a Multinomial$(N,p_1,\dots,p_{m_0})$ distribution, the number of data generated locally. In this way, for the $k$-th device, the size of its local dataset is, on average, proportional to $Np_k$ where $N$ is the size of the training set and $p_k$ the probability of assigning the $n$-th data point of the dataset to device $k$. To obtain class heterogeneity, each device picks and stores in the local cache a sample from the positive class with probability $p_+$  and a negative one with probability $1-p_+$ (remember that in our simulation we consider a binary classification task, thus samples are either of the positive class or of the negative class).  We randomly assign a different $p_+$ to each device, i.e., $p_+$ is sampled from a uniform distribution in $[0.15, 0.85]$ to guarantee at least $0.15$ probability to either class.

For each grouping $\gamma$, we then select the CPs to guarantee heterogeneity of data also after collection at CPs. First, we select CPs as the nodes with the most imbalanced class representation, where imbalance is computed as the entropy over the class distribution of the local dataset. Moreover, nodes that are not CPs move their data according to a \emph{per class preferential attachment} procedure. Precisely, a non-CP device makes distinct CP selections, one for each of the classes contained in its local dataset. The probability to select a CP for a given class depends on the number of samples contained in its local dataset w.r.t. the other CPs, i.e., the more samples of a given class are held by a CP, the higher is its probability of being selected. This approach to the definition of CPs and selection of CPs by other nodes is quite unlikely in reality. However, it guarantees large heterogeneity of the resulting distribution of data at CPs, stressing the validation of the analytical model. Cases characterised by lower heterogeneity of data at CPs show similar quantitative results, not shown here for space reasons. Note that, for space reasons, we focus on data heterogeneity only. However, results also indicate the robustness of the model to other types of heterogeneity, namely device heterogeneity. Our configurations result in large imbalance in data at Collection Points, which leads to heterogeneous network and computation costs. This is the same outcome of considering heterogeneous devices in terms of communication or computing characteristics. Figure\ref{fig:data_init} shows an example of heterogeneous data distribution among devices, before and after the collection process. 

\begin{figure}
    \centering
    \subfloat[Before aggregation on CPs]{\includegraphics[width=.5\columnwidth]{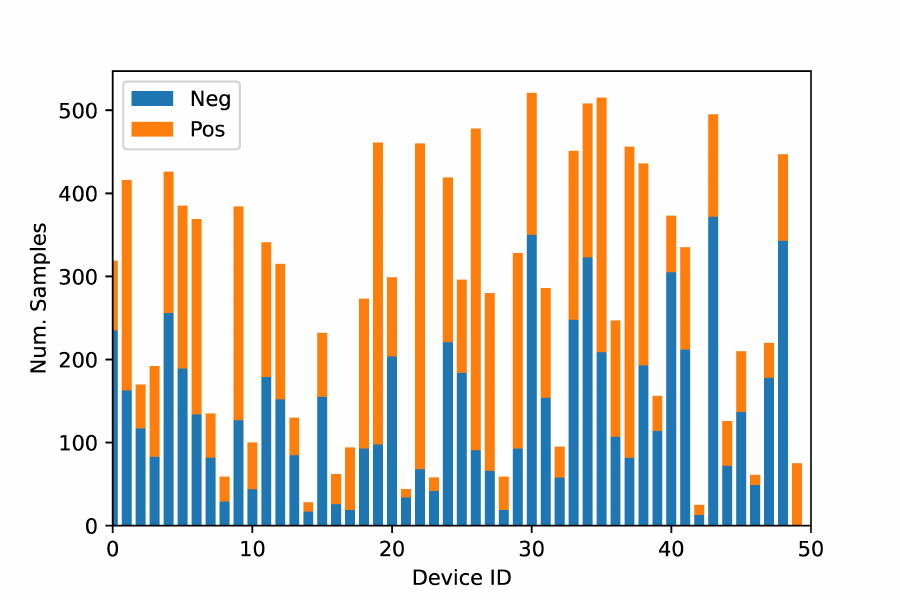}\label{fig:before}}
    \subfloat[After aggregation on CPs]{\includegraphics[width=.5\columnwidth]{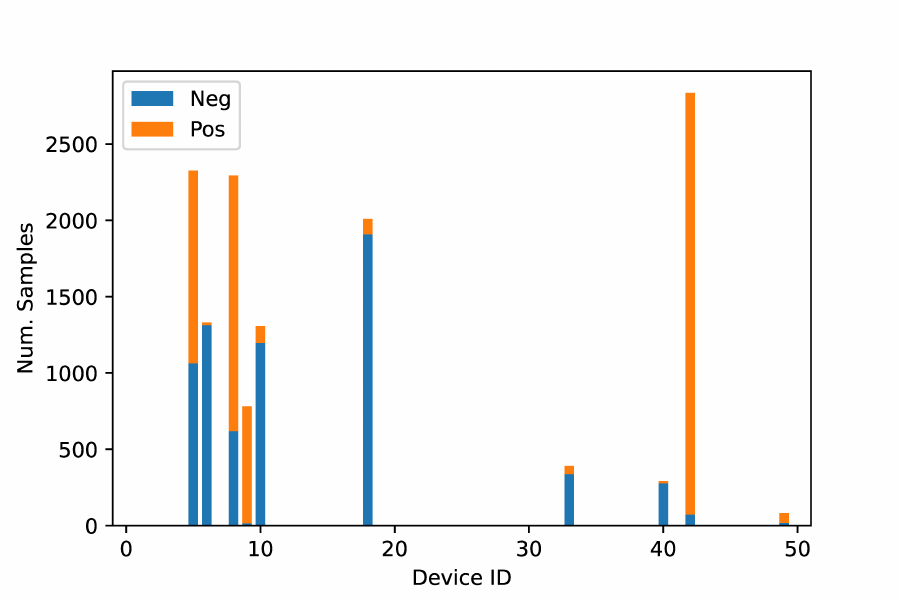}\label{fig:after}}
    \caption{Example of heterogeneous data distribution before (a) and after data aggregation on 10 CPs.}
    \label{fig:data_init}
\end{figure}

\section{Performance analysis}
\label{sec:perf}
In this section, we assess the accuracy of the model in estimating the optimal operating point across the entire range of possible aggregation levels (Section~\ref{sub:validation}) and we provide an extensive performance analysis to compare the cost of a system configured to operate at the optimal point estimated by our model, as compared to full centralisation and full decentralisation (Section \ref{sub:comparison}). Finally, in Section \ref{sub:sensitivity} we study the sensitivity of our model to its parameters. 
\begin{table}[t!]
\caption{Validation of the model with linear $f(\gamma)$ varying $\varepsilon$, for each dataset (DS).}
\begin{center}
\scriptsize
\begin{tabular}{lllllllll}
\toprule
 DS&$\varepsilon$ & $\gamma^*$ & $\hat\gamma$ & $R(\gamma^*)$ & $R(\hat\gamma)$ & $C(\gamma^*)$  & $C(\hat\gamma)$ & OH \\
 && & & & & (Eu$\cdot10^6$) & (Eu$\cdot10^6$) & (\%)\\
 \midrule
\multirow{5}{*}{CT}
& $10^{-7}$ & 6 & 7 & 34 & 31 & 1.91 & 1.93 & 1.10 \\
& $10^{-6}$ & 6 & 7 & 29 & 27 & 1.82 & 1.86 & 1.93 \\
& $10^{-5}$ & 6 & 6 & 24 & 24 & 1.76 & 1.76 & 0.00 \\
& $10^{-4}$ & 5 & 6 & 22 & 19 & 1.66 & 1.68 & 1.47 \\
& $10^{-3}$ & 1 & 5 & 24 & 10 & 1.05 & 1.51 & 44.66 \\
\addlinespace
\multirow{5}{*}{GS}
& $10^{-7}$ & 7 & 7 & 45 & 45 & 93.48 & 93.48 & 0.00 \\
& $10^{-6}$ & 6 & 7 & 45 & 39 & 74.75 & 76.01 & 1.69 \\
& $10^{-5}$ & 7 & 7 & 32 & 32 & 60.55 & 60.55 & 0.00 \\
& $10^{-4}$ & 6 & 7 & 30 & 26 & 47.51 & 48.33 & 1.73 \\
& $10^{-3}$ & 7 & 6 & 19 & 23 & 37.36 & 38.16 & 2.14 \\
\addlinespace
\multirow{5}{*}{MN}
&$10^{-7}$ & 3 & 3 & 58 & 58 & 12.06 & 12.06 & 0.00 \\
&$10^{-6}$ & 3 & 3 & 50 & 50 & 10.68 & 10.68 & 0.00 \\
&$10^{-5}$ & 3 & 3 & 42 & 42 & 9.27 & 9.27 & 0.00 \\
&$10^{-4}$ & 3 & 2 & 33 & 58 & 8.14 & 8.47 & 4.07 \\
&$10^{-3}$ & 2 & 2 & 43 & 43 & 7.09 & 7.09 & 0.00 \\
\bottomrule
\end{tabular} 
\end{center}
\label{tab:linear_valid}
\end{table}
\subsection{Model Validation}
\label{sub:validation}

\begin{figure}[t]
	\centering
	\subfloat[CT $\varepsilon=10^{-3}$]{\includegraphics[width=.50\columnwidth]{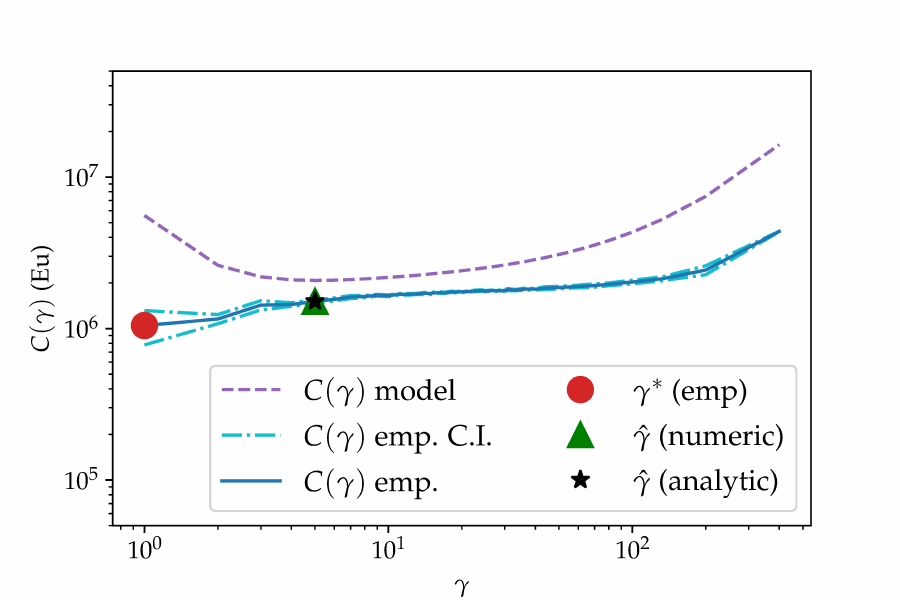}\label{fig:lin_bad}}
	\subfloat[CT $\varepsilon=10^{-7}$]{\includegraphics[width=.50\columnwidth]{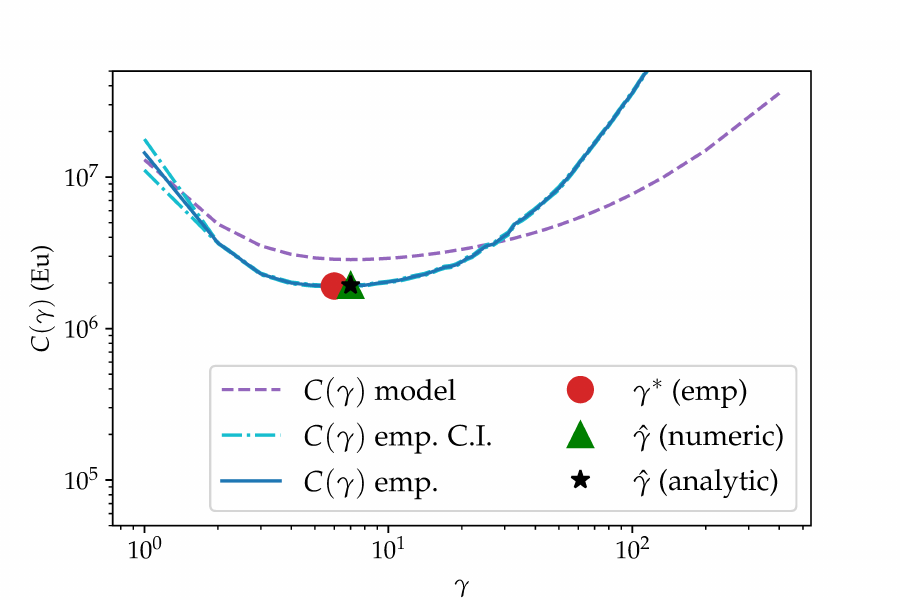}\label{fig:lin_good}}\\
	\subfloat[MN $\varepsilon=10^{-3}$]{\includegraphics[width=.50\columnwidth]{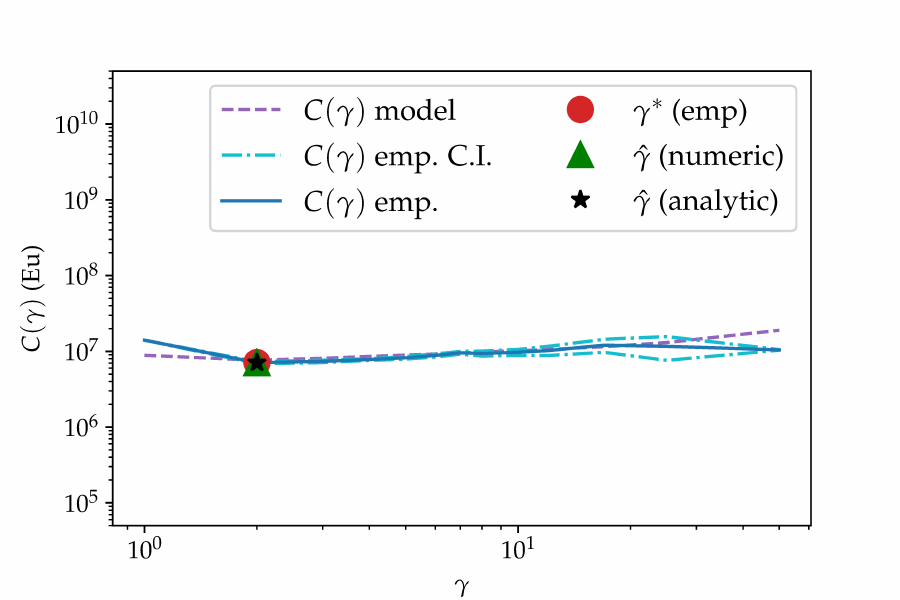}}
	\subfloat[MN $\varepsilon=10^{-7}$]{\includegraphics[width=.50\columnwidth]{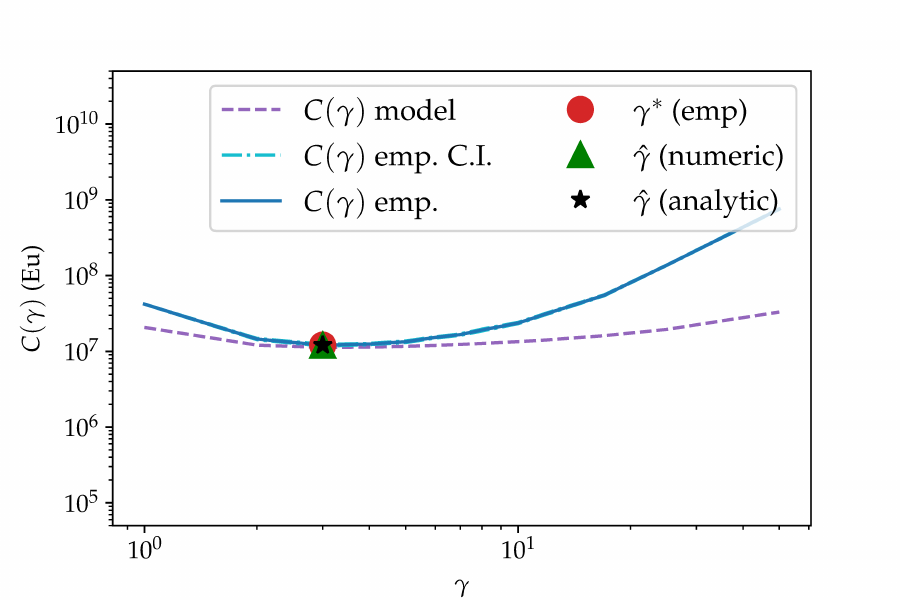}}
	\caption{Linear computational cost ($\alpha=1$) for Covtype (CT) and MNIST (MN) datasets and two target accuracies ($\varepsilon=10^{-3},10^{-7}$). Empirical (blue line) and model's predicted (dashed purple) cost $C(\gamma)$ at different aggregation level ($\gamma$) . Plot reports also the optimal empirical operational point (red dot) and the one identified by our model analytically (black star) and numerically (green triangle).}
	\label{fig:lin_curve}
\end{figure}    
In this section, we validate our model for three configurations of computational cost. We begin analysing the linear computational cost, corresponding to a baseline case where the compute cost is proportional to the number of data points handled by each collection point. In the following we refer to Table~\ref{tab:linear_valid} and Fig.~\ref{fig:lin_curve}. In the table we report, for each target accuracy and for each dataset, the empirical optimal aggregation point ($\gamma^*$) the one found by our model ($\hat\gamma$), the number of rounds actually needed to reach the target accuracy ($R(\gamma^*)$) and those induced by configuring the system at the estimated optimal operating point ($R(\hat\gamma)$). For each solution, we also report the optimal empirical cost ($C(\gamma^*)$) and the one corresponding to the estimated optimal aggregation ($C(\hat\gamma)$). Finally, we provide the percentage of additional cost due to the approximation of the model, i.e., the \emph{overhead} (OH). 
Fig.~\ref{fig:lin_curve} focuses on the linear computation costs ($\alpha=1$) and shows, for the extreme values of $\varepsilon$ in the range $[10^{-3},10^{-7}]$ and for a subset of datasets, the curves of the empirical cost function $C(\gamma)$ (solid blue line) and the one calculated by our model (dashed purple line), the optimal empirical aggregation point (red circle), the optimal operating point found by solving numerically our model (green triangle) and the one found though the closed-form solution (black star).

\begin{table}[ht!]
\caption{Validation of the model with sub-linear $f(\gamma)$ varying $\varepsilon$, for each dataset (DS).}
\begin{center}
\scriptsize
\begin{tabular}{lllllllll}
  \toprule
   DS&$\varepsilon$ & $\gamma^*$ & $\hat\gamma$ & $R(\gamma^*)$ & $R(\hat\gamma)$ & $C(\gamma^*)$  & $C(\hat\gamma)$ & OH \\
   && & & & & (Eu$\cdot10^6$) & (Eu$\cdot10^6$) & (\%)\\
   \midrule
  \multirow{5}{*}{CT}
  & $10^{-7}$ & 7 & 17 & 31 & 25 & 1.81 & 1.95 & 7.47 \\
  & $10^{-6}$ & 6 & 15 & 29 & 21 & 1.76 & 1.84 & 4.43 \\
  & $10^{-5}$ & 6 & 13 & 24 & 18 & 1.71 & 1.79 & 4.41 \\
  & $10^{-4}$ & 5 & 11 & 22 & 15 & 1.66 & 1.74 & 4.79 \\
  & $10^{-3}$ & 1 & 8 & 24 & 8 & 1.05 & 1.63 & 55.53 \\
  \addlinespace
  \multirow{5}{*}{GS}
  & $10^{-7}$ & 7 & 19 & 45 & 25 & 51.98 & 73.13 & 40.71 \\
  & $10^{-6}$ & 8 & 18 & 34 & 21 & 45.33 & 60.63 & 33.75 \\
  & $10^{-5}$ & 8 & 17 & 28 & 18 & 38.95 & 49.82 & 27.93 \\
  & $10^{-4}$ & 7 & 16 & 26 & 16 & 34.08 & 43.09 & 26.45 \\
  & $10^{-3}$ & 7 & 14 & 19 & 12 & 29.87 & 32.16 & 7.66 \\
  \addlinespace
  \multirow{5}{*}{MN}
  &$10^{-7}$ & 4 & 5 & 42 & 26 & 9.97 & 10.14 & 1.77\\
  &$10^{-6}$ & 3 & 4 & 50 & 22 & 9.13 &  9.16 & 0.40\\
  &$10^{-5}$ & 3 & 4 & 42 & 18 & 8.30 &  8.47 & 2.13\\
  &$10^{-4}$ & 3 & 3 & 33 & 33 & 7.49 &  7.49 & 0.00\\
  &$10^{-3}$ & 2 & 2 & 43 & 43 & 6.80 &  6.80 & 0.00\\
  \bottomrule
  \end{tabular}  
\end{center}
\label{tab:sublinear_valid}
\end{table}

In all datasets, the estimation of the aggregation level is equal to or greater than the optimal one ($\gamma^*$). Interestingly,  considering the linear computational cost, the overheads (OH) induced by the overestimation are limited, especially when the accuracy requested is medium/high, i.e., in all the datasets and for $\varepsilon\leq10^{-4}$ the overhead is up to 4.07\%. We justify this performance by observing in Fig.~\ref{fig:lin_curve} that, a part from a specific corner case (CT, $\varepsilon=10^{-3}$), our model's cost function closely approximates the empirical one especially nearby the optimal point. In those cases where our model overestimates the optimal solution, the empirical curve is flat enough such that the overhead remains contained. As anticipated, the only corner case regards the Covtype dataset for $\varepsilon=10^{-3}$ where our model's estimation incurs in high overhead (44.66\%). As revealed in Figs.~\ref{fig:lin_bad}-\ref{fig:lin_good}, the inaccuracy of our model is due to the fact that the empirical cost function in the first trait has two distinct regimes: for a lower accuracy (e.g., $\varepsilon\geq10^{-3}$) the curve is \emph{concave-down} while for a higher accuracy ($\varepsilon\leq10^{-7}$) it is \emph{concave-up}. The model switches from one regime to the other too soon, thus causing the imprecise approximation. 

This initial set of results (Fig.~\ref{fig:lin_curve} and Table~\ref{tab:linear_valid}) already shows that our model is quite accurate in estimating the optimal operating point. Interestingly, unless in the specific case of $\varepsilon=10^{-3}$, the error in estimating the optimal operating point leads to marginal additional costs (below 5\%). Moreover, the prediction accuracy improves with the targeted $\varepsilon$-accuracy, which is a very positive feature.

We present now the simulation results related to a computational cost that grows non-linearly with the amount of data to be processed on a device. Remember that sub-linear cases represent configurations where ``it is cheap to aggregate", such as edge infrastructure operators implementing cost policies to push moving data on their devices. Super-linear cases represent configurations where ``it is costly to aggregate", such as when collection points are resource constrained devices which risk saturation. In our model, super-linear and sub-linear cases correspond to values of $\alpha$ greater and lower than 1, respectively. In our simulations we set $\alpha=2$ and  $\alpha=0.5$ as representative of the super-linear and sub-linear cases, respectively. Results are shown in Table~\ref{tab:sublinear_valid} (sub-linear case), and Table~\ref{tab:superlinear_valid} (super-linear case).
\begin{table}[ht!]
\caption{Validation of the model with super-linear $f(\gamma)$ varying $\varepsilon$, for each dataset (DS).}
\begin{center}
\scriptsize
\begin{tabular}{lllllllll}
  \toprule
   DS&$\varepsilon$ & $\gamma^*$ & $\hat\gamma$ & $R(\gamma^*)$ & $R(\hat\gamma)$ & $C(\gamma^*)$  & $C(\hat\gamma)$ & OH \\
   && & & & & (Eu$\cdot10^6$) & (Eu$\cdot10^6$) & (\%)\\
   \midrule
  \multirow{5}{*}{CT}
  & $10^{-7}$ & 4 & 3 & 47 & 66 & 2.44 & 2.64 & 7.88 \\
  & $10^{-6}$ & 4 & 3 & 41 & 56 & 2.22 & 2.36 & 6.44 \\
  & $10^{-5}$ & 4 & 3 & 34 & 47 & 2.02 & 2.13 & 5.11 \\
  & $10^{-4}$ & 4 & 3 & 27 & 44 & 1.85 & 1.94 & 5.34 \\
  & $10^{-3}$ & 1 & 3 & 24 & 19 & 1.05 & 1.49 & 42.09\\
 \addlinespace 
  \multirow{5}{*}{GS}
  & $10^{-7}$ & 4 & 3 & 90 & 149 & 333.11 & 351.79 & 5.61 \\
  & $10^{-6}$ & 4 & 3 & 76 & 127 & 251.73 & 268.14 & 6.52 \\
  & $10^{-5}$ & 4 & 3 & 65 & 106 & 185.89 & 194.85 & 4.82 \\
  & $10^{-4}$ & 4 & 3 & 51 &  86 & 126.12 & 134.91 & 6.97 \\
  & $10^{-3}$ & 4 & 3 & 38 &  63 &  81.94 &  85.84 & 4.76 \\
  \addlinespace
  \multirow{5}{*}{MN}
  &$10^{-7}$ & 2 & 2 & 101 & 101 & 21.00 & 21.00 & 0.00 \\
  &$10^{-6}$ & 2 & 2 &  86 &  86 & 17.09 & 17.09 & 0.00 \\
  &$10^{-5}$ & 2 & 2 &  72 &  72 & 13.75 & 13.75 & 0.00 \\
  &$10^{-4}$ & 2 & 2 & 58 &  58 & 10.85 & 10.85 & 0.00 \\
  &$10^{-3}$ & 2 & 1 &  43 & 143 & 8.32 & 14.05 & 68.87 \\
  \bottomrule
  \end{tabular}
\end{center}
\label{tab:superlinear_valid}
\end{table}

We focus first on the case when the computational cost is sub-linear.
Results related to the CovType(CT) and MNIST (MN) dataset confirm the same behaviour already observed in the linear case. The case of Gisette (GS) is instead quite the opposite. The model predictions are not very accurate in this case, which leads to a significant overhead. In addition, the overhead in this case is higher for higher accuracy. We remark that this is the only case among the ones we tested where the model is not able to provide very good approximation levels (unless for the corner cases in CT already discussed). This is a side effect of the specificity of the dataset, as well as the simplifying assumptions required for the model. We believe, however, that the high accuracy shown in the rest of the cases (also with the same dataset) confirms the validity of our approach. We anticipate that, as we show in Section~\ref{sub:comparison}, in spite of the accuracy degradation observed in this case, the operational points identified by our model allows to configure the system in a way to save significant resources (up to 93.64\%) w.r.t. the cost connected to naive configurations such as full decentralisation or centralisation.  

For the super-linear case, at each $\varepsilon$, the optimal operating point is at lower aggregation levels with respect to the linear and sub-linear cases. This is expected, because the steep growth of the computational cost pushes the optimal aggregation point towards more decentralised solutions (lower values of $\gamma$). From the system point of view, this means that the distributed learning process performs more communication rounds. This is confirmed by the number of rounds performed at, for example, $\varepsilon=10^{-6},10^{-7}$ (see columns $R(\gamma^*),R(\hat\gamma)$ of Table~\ref{tab:superlinear_valid} compared to the corresponding ones in Table~\ref{tab:sublinear_valid}). 
In terms of accuracy of the model we observe that, for all datasets,  at lower accuracy ($\varepsilon=10^{-3}$) the model estimation error results in the higher overhead (as in the previous cases) while for the all the remaining cases the additional cost is quite limited, i.e., are up to 8\%.

Summarising this set of results, we can observe some common features across the three cases, which highlights a coherent behaviour of our model, unless for the only one case of GS with sub-linear computation costs. First, the model tends to be more accurate in estimating the costs as the required accuracy increases. Moreover, in many cases the model is able to estimate the exact optimal operating point, i.e., the best level of aggregation that drives the system to achieve minimal cost. Finally, even when the optimal operating point is only approximated, the incurred overhead is most of the time quite limited, with only a very few exceptions. We can thus conclude that the model is an accurate tool to both (i) analyse the cost incurred at the different aggregation levels, and (ii) identify optimal operation points of the decentralised learning task.

Finally, to assess the quality of the final training model trained by the distributed learning we  report in Fig.~\ref{fig:accuracy} its final accuracy at each value of $\gamma$, for all the considered datasets. For the Covtype and MNIST datasets the inference accuracy on the test set is almost constant for the aggregation level. This means that, aggregating data at different levels does not or only marginally affect the accuracy of the trained model. For the case of Gisette, the difference between full decentralisation and full decentralisation is below 5\%. This is the maximum lack of accuracy we obtain by using any decentralised configuration, including the one indicated by our model.

\begin{figure}[ht]
	\centering
	\subfloat[CT]{\includegraphics[width=.50\columnwidth]{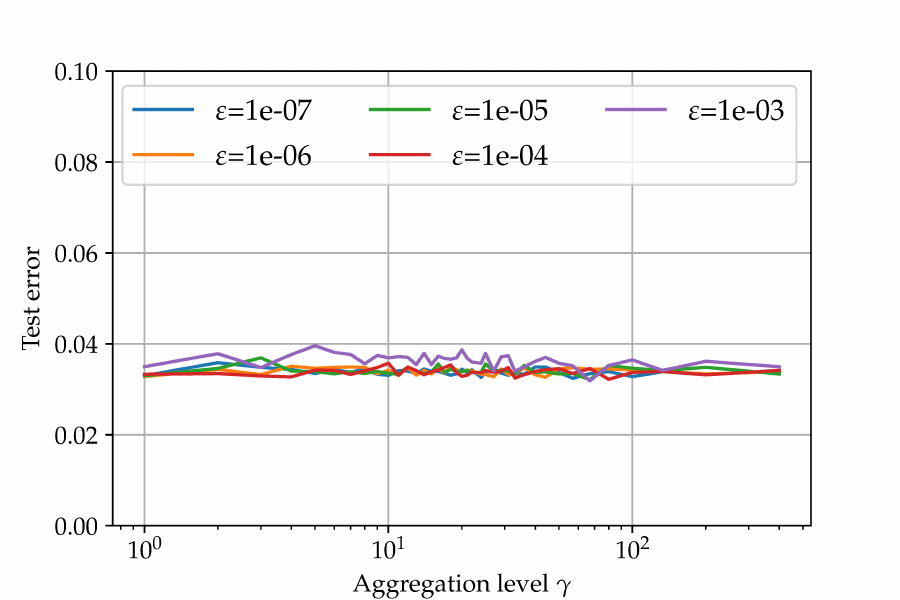}}\\
	\subfloat[GS]{\includegraphics[width=.50\columnwidth]{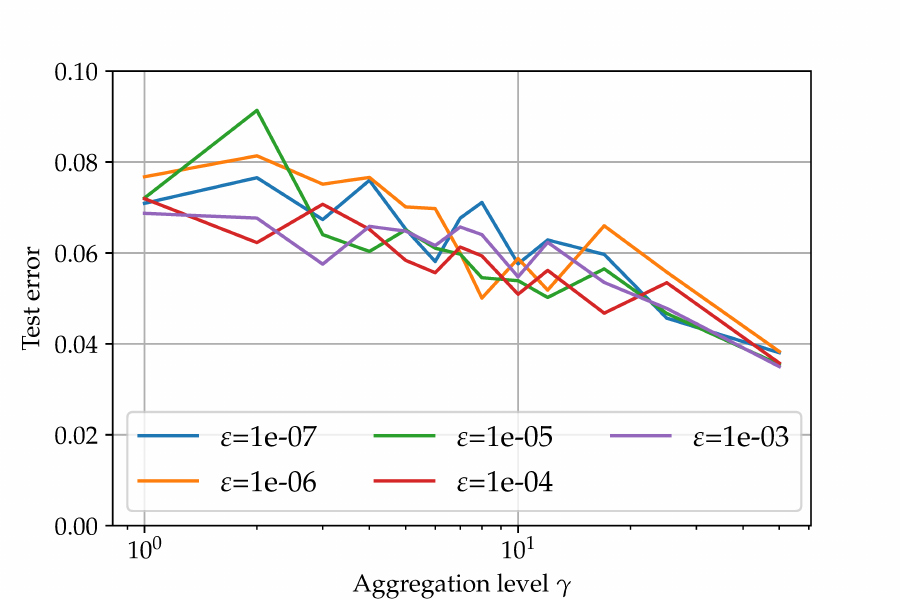}}
	\subfloat[MN]{\includegraphics[width=.50\columnwidth]{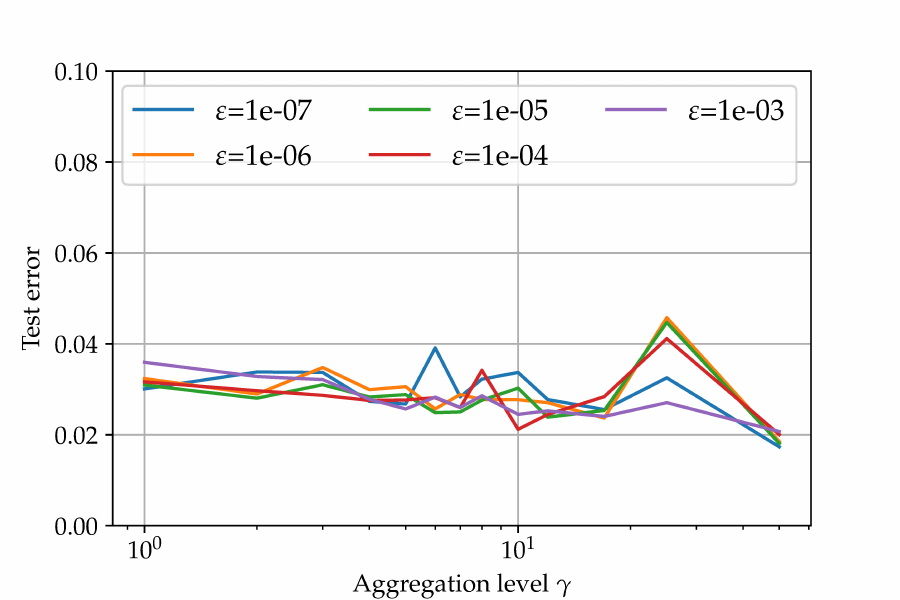}}
	\caption{Test error on all the datasets, for different target accuracy, at each aggregation level $\gamma$.}
	\label{fig:accuracy}
\end{figure}

\subsection{Comparison of optimal, fully centralised and fully decentralised operating points}
\label{sub:comparison}
An important feature to be noted is that, in most of the analysed cases, the optimal operating point is an \emph{intermediate} aggregation level, i.e., the optimal configuration is neither to fully centralise ($\gamma=m_0$), neither to fully decentralise ($\gamma=1)$ the learning task. It is thus interesting to quantify the advantage of using our model and configure the system at the predicted optimal operating point, with respect to applying simpler, but less accurate, policies corresponding to the two extreme aggregation configurations. For this analysis we refer to Table~\ref{tab:all_costs}, where we report the results for the linear ($\alpha=1$), sub-linear ($\alpha=0.5$) and super-linear ($\alpha=2$) cases of the Covtype dataset, respectively. Due to space reasons we limit the discussion to this case because it is representative also for the other two. For the sake of completeness, we report in Tables~\ref{tab:centr_gain},~\ref{tab:dec_gain} the gain/loss (\%) deriving by the usage of the model's intermediate solution in place of the fully distributed and fully centralised ones, for all the datases and considered cases.
\begin{table}[ht!]
\centering
\caption{Costs at the predicted optimal operating point compared to those of full decentralisation and centralisation. Covtype dataset (CT) }
\label{tab:all_costs}

\begin{center}
\scriptsize
\begin{tabular}{@{}cccccc@{}}
\toprule
$\alpha$ & $\varepsilon$ & $\widehat{\gamma}$ & $C(1)$& $C(\widehat{\gamma})$ &  $C(400)$\\ 
& & & (Eu$\cdot 10^6$) & (Eu$\cdot 10^6$) & (Eu$\cdot 10^6$) \\ 
\midrule
\multirow{6}{*}{$0.5$} 
& $10^{-7}$ & 17 & 14.44 & 1.95 & 63.79  \\
& $10^{-6}$ & 15 & 14.53 & 1.84 & 48.54  \\
& $10^{-5}$ & 13 & 5.85 & 1.79 & 16.26  \\
& $10^{-4}$ & 11 & 2.82 & 1.74 & 2.00  \\
& $10^{-3}$ & 8 & 1.05 & 1.63 & 1.93  \\
\addlinespace
\multirow{6}{*}{$1$} 
& $10^{-7}$ & 7 & 14.44 & 1.93 & 1241.61  \\
& $10^{-6}$ & 7 & 14.53 & 1.86 &  936.48  \\
& $10^{-5}$ & 6 &  5.85 & 1.76 &  291.04  \\
& $10^{-4}$ & 6 &  2.82 & 1.68 &    5.74  \\
& $10^{-3}$ & 5 &  1.05 & 1.51 & 4.37  \\
\addlinespace
\multirow{6}{*}{$2$} 
& $10^{-7}$ & 3 & 14.44 & 2.64 & 495930.70  \\
& $10^{-6}$ & 3 & 14.53 & 2.36 & 373880.88  \\
& $10^{-5}$ & 3 & 5.85 & 2.13 & 115700.96  \\
& $10^{-4}$ & 3 & 2.82 & 1.94 & 1585.89  \\
& $10^{-3}$ & 3 & 1.05 & 1.49 & 1035.64  \\
\bottomrule
\end{tabular}
\end{center}
\end{table}
\begin{table}[ht!]
  \centering
  \caption{Gain/loss obtained from comparing the cost at the predicted optimal operating point compared to those of full centralisation, for each value of $\varepsilon$ and each dataset (DS) }
  \label{tab:centr_gain}
  \begin{center}
  \scriptsize
  \begin{tabular}{lllllll}
  \toprule
  DS & $\alpha$ & \multicolumn{5}{c}{Gain/loss(\%)}\\
  \cmidrule{3-7}
    & & $\varepsilon=10^{-7}$ & $\varepsilon=10^{-6}$ & $\varepsilon=10^{-5}$ &  $\varepsilon=10^{-4}$ & $\varepsilon=10^{-3}$  \\
  \midrule
  \multirow{3}{*}{CT} 
  & $0.5$ & 96.95 & 96.20 & 89.00 & 13.08 & 15.67 \\  
  & $1$ & 99.84 & 99.80 & 99.40 & 70.70 & 65.39 \\
  & $2$ & 100.00 & 100.00 & 100.00 & 99.88 & 99.86\\
\addlinespace
\multirow{3}{*}{GS} 
  & $0.5$ & 81.72 & 78.70 & 76.78 & 72.35 & 65.47\\  
  & $1$ & 96.49 & 95.88 & 95.52 & 94.83 & 92.24 \\
  & $2$ & 99.73 & 99.71 & 99.71 & 99.70 & 99.63 \\
  \addlinespace
  \multirow{3}{*}{MN} 
  & $0.5$ & 91.10 & 89.65 & 87.27 & 82.42 & 22.44 \\  
  & $1$ & 98.40 & 98.14 & 97.79 & 96.61 & 32.17 \\
  & $2$ & 99.94 & 99.94 & 99.93 & 99.91 & 86.28\\
  \bottomrule
  \end{tabular}
  \end{center}
  \end{table}

  \begin{table}[ht!]
    \centering
    \caption{Gain/loss obtained from comparing the cost at the predicted optimal operating point compared to those of full decentralisation, for each value of $\varepsilon$ and each dataset (DS)}
    \label{tab:dec_gain}
    
    \begin{center}
    \scriptsize
    \begin{tabular}{lllllll}
    \toprule
    DS & $\alpha$ & \multicolumn{5}{c}{Gain/loss(\%)}\\
    \cmidrule{3-7}
    & & $\varepsilon=10^{-7}$ & $\varepsilon=10^{-6}$ & $\varepsilon=10^{-5}$ &  $\varepsilon=10^{-4}$ & $\varepsilon=10^{-3}$  \\
    \midrule
    \multirow{3}{*}{CT} 
    & $0.5$ &  86.52 & 87.32 & 69.45 & 38.45 & -55.53 \\  
    & $1$ & 86.63 & 87.21 & 70.02 & 40.42 & -44.66\\
    & $2$ & 81.74 & 83.73 & 63.67 & 31.08 & -42.09\\
  \addlinespace
  \multirow{3}{*}{GS} 
    & $0.5$ & 93.64 & 93.35 & 92.86 & 91.55 & 90.70 \\  
    & $1$ & 91.87 & 91.66 & 91.32 & 90.52 & 88.97 \\
    & $2$ & 69.39 & 70.59 & 72.08 & 73.54 & 75.19 \\
    \addlinespace
    \multirow{3}{*}{MN} 
    & $0.5$ & 75.84 & 73.14 & 68.27 & 62.53 & 51.63  \\  
    & $1$ &71.28 & 68.71 & 65.28 & 57.62 & 49.52 \\
    & $2$ & 49.98 & 49.90 & 48.52 & 45.68 & 0.00 \\
    \bottomrule
    \end{tabular}
    \end{center}
    \end{table}

Table~\ref{tab:all_costs} shows that the cost at the operating point indicated by our model (column $C(\widehat\gamma)$) is far lower than both the one of a fully centralised and fully decentralised system (columns $C(400)$ and $C(1)$, respectively). Moreover, it sheds light on the reasons why moving from less to more expensive computation pushes the optimal aggregation point towards higher decentralisation. Precisely, looking at column ($C(400)$) of the table, the costs connected to the full centralisation increase up to 4 order of magnitude (e.g. $\varepsilon=10^{-6},10^{-7}$) when we move form the sub-linear to the super-linear regimes. Consequently, the optimal operating point reduces accordingly, doubling the optimal number of collection points in the super-linear regime for the same level of accuracy (i.e., $\varepsilon=10^{-7}$). 

The advantage of using the intermediate  solution instead of centralising all the data on one device is confirmed looking at the gain figures in Table~\ref{tab:centr_gain}.
In the sub-linear computational cost case, we can save form $13\%$ to $97\%$ Eu, according to the target accuracy of interest and the specific dataset. The highest gain corresponds to the highest accuracy ($\varepsilon=10^{-7}$), since the empirical cost function tends to be quadratic and, thus, requiring a higher accuracy acts as a scale factor pushing its extremes (both $\gamma=1$ and $\gamma=m_0$) further and further from the optimal intermediate point.  From the system point of view, this means that as we aggregate on fewer collection points, the number of data points per CP increases along with the local computation which adds up significantly over the network cost. 
Using ``more aggressive'' cost functions (linear and super-linear) magnifies such a behaviour. The costs for the full centralisation increases dramatically and, thus, configuring the ML al indicated by our model results in gains ranging from 32\% up to 100\%.

In comparison with full decentralisation the gain obtainable by using the aggregation level provided by our model ranges from $0\%$ to $93.64\%$ (apart from the corner cases of low accuracy in the CovType dataset already discussed). Overall, looking at Table~\ref{tab:dec_gain} we see that  the different  regimes of computational cost follow the same decreasing trend observed in Table\ref{tab:centr_gain}, i.e., the highest accuracy the highest the gain.  This is due the fact that the computational costs increases with the accuracy required, contributing more to the overall cost. Since full decentralisation is the configuration in which devices have the least quantity of local data, this configuration forces the algorithm to perform more communication rounds (and therefore also more local computation) which might be avoided with an intermediate level of data aggregation. Finally, note that using our model results in a gain also in the Gisette dataset where the model provides a poor approximation of the real optimal operating point. This shows that even in such case, our approach is significantly better that naive solutions such as full decentralisation or centralisation.

\subsection{Sensitivity analysis}
\label{sub:sensitivity}
We finally present a sensitivity analysis of the model with respect to key parameters. Note that, for this analysis we adopt the same configuration used during the validation, with the aim of making the results easier to interpret.   
By default, we set the parameters as already indicated in Table~\ref{tab:settings}. We fix the target accuracy $\varepsilon=10^{-5}$ and we consider that the input data have dimensionality $d=54$ as in the Covtype dataset. The unitary cost for networking is $\theta=1 Eu$, and we set the $\mu$ parameter to $10^{-4}$. Therefore, the computational unitary cost is $\beta=10^{-4}$ Eu. We show the results for all the considered computational cost regimes: sub-linear, linear and super-linear which correspond to $\alpha=0.5, \alpha=1, \alpha=2$, respectively. Unless differently specified, we assume to have $m_0=400$ devices, each one holding $n_0=112$ data points.

First, we analyse the sensitivity of our model to increasingly higher amounts of data generated at each device ($n_0$). Results in Table~\ref{tab:n_var} show that, for all the regimes of $\alpha$, our model tends to aggregate the data on fewer collection points when at each device the amount of data is smaller. Therefore, the optimal number of collection points is directly proportional to the initial quantity of data per device. The reason behind this behaviour is that, when $n_0$ increases, the local models can be trained with more data, and thus are more accurate. Therefore, fewer rounds are required, and thus the networking costs at a given aggregation level decreases. This pushes towards more decentralised optimal operating points. As expected, for each value of $n_0$ the optimal operating point moves towards more collection points as the computational cost grows faster with $\gamma$. 

\begin{table}
\centering
\caption{Optimal operational points varying the amount of data generated per device.}
\label{tab:n_var}
\scriptsize
\begin{tabular}{cccc}
\toprule
$n_0$ & $\widehat{\gamma}_{\alpha=0.5}$  & $\widehat{\gamma}_{\alpha=1}$  &  $\widehat{\gamma}_{\alpha=2}$  \\
& (sub-linear) & (linear) & (super-linear)\\
\midrule
$50$ 	& $11$ 	& $6$ 	& $3$ \\
$75$ 	& $6$ 	& $4$ 	& $2$ \\
$100$ 	& $3$ 	& $3$ 	& $2$ \\
$125$ 	& $2$ 	& $2$ 	& $1$ \\
$200$ 	& $1$ 	& $1$ 	& $1$ \\
$500$ 	& $1$ 	& $1$ 	& $1$ \\
$1000$ 	& $1$ 	& $1$ 	& $1$ \\
\bottomrule
\end{tabular}
\end{table}
\begin{table}[ht]
\centering
\caption{Optimal operational points varying the initial amount of data per device.}
\label{tab:m_var}
\scriptsize
\begin{tabular}{cccc}
\toprule
$m_0$ & $\widehat{\gamma}_{\alpha=0.5}$  & $\widehat{\gamma}_{\alpha=1}$  &  $\widehat{\gamma}_{\alpha=2}$  \\
& (sub-linear) & (linear) & (super-linear)\\
\midrule
$50$ 	& $1$ 	& $1$ 	& $1$ \\
$75$ 	& $1$ 	& $1$ 	& $1$ \\
$100$ 	& $1$ 	& $1$ 	& $1$ \\
$125$ 	& $2$ 	& $1$ 	& $1$ \\
$200$ 	& $2$ 	& $2$ 	& $1$ \\
$500$ 	& $3$ 	& $3$ 	& $2$ \\
$1000$ 	& $5$ 	& $3$ 	& $2$ \\
$2000$ 	& $7$ 	& $4$ 	& $2$ \\
$5000$ 	& $12$ 	& $6$ 	& $2$ \\
\bottomrule
\end{tabular}
\end{table}
We then analyse the effect of the number of devices, $m_0$, by keeping the data they individually generate ($n_0$) fixed. As it was the case when increasing $n_0$, also by increasing $m_0$ we are increasing the total amount of data in the system. However, looking at results in Table~\ref{tab:m_var} we notice that in the latter case, the optimal behaviour is to aggregate \emph{(slightly) more} (rather than less) as the total amount of data increases. 

The reason is that, when increasing the number of nodes at a given aggregation level, we are keeping constant the number of nodes that contribute their data to an individual Collection Point (remember that $\gamma$ has exactly this physical meaning), and we are increasing the number of Collection Points. Therefore, the amount of data over which Collection Points compute local models are the same, which means the accuracy of the local models are similar. However, with more Collection Points the network traffic per round of model update increases. Therefore, it is more efficient to aggregate data on fewer Collection Points, paying the initial cost of moving more raw data, but saving on the network traffic during collective training.

This is an interesting result as it shows that variations in the total amount of data in the system may generate opposite variations in terms of optimal operating point, depending on whether this is a side effect of a higher amount of data generated at each individual device, or a higher number of devices collaborating in the system.

Finally we analyse the impact of the relationship between the communication and computing costs. To this end we vary the $\mu$ parameter, which controls the ratio between $\theta$ and $\beta$, in the range $[10^{-5},10]$. In this way we are considering both scenarios in which the computational cost per device is negligible (low values of $\mu$) and others where they exceed the communication costs (values of $\mu$ higher than 1). We jointly analyse the cases where the amount of data generated at each device increases. We perform five blocks of simulations. For each block we keep fixed the number of devices ($m_0=400$) and we increase the amount of data they individually generate ($n_0=50,100,200$). For each pair $(m_0,n_0)$ we vary the parameter $\mu$. Note that, considering cases where the computation cost is negligible allows us to also include in the analysis cases where computation is offered ``for free". Indeed, in such cases we found that the computation component of the cost function is orders of magnitude lower than the network component. Even in such cases, our results show that aggregating all data in a unique centralisation point is not the optimal choice. Results are reported in Table~\ref{tab:mu_var}.
\begin{table}[ht!]
\centering
\caption{Optimal operational points varying the amount of data per device and the ratio between the communication and the computation cost.}
\label{tab:mu_var}
\scriptsize
\begin{tabular}{llccc}
\toprule
$n_0$ & $\mu$ & $\widehat{\gamma}_{\alpha=0.5}$ & $\widehat{\gamma}_{\alpha=1}$  &  $\widehat{\gamma}_{\alpha=2}$ \\
& & (sub-linear) & (linear) & (super-linear)\\
\midrule
\multirow{7}{*}{$50$}  & $10$ & $7$ & $1$ & $1$ \\
  & $1$ & $7$ & $1$ & $1$ \\
  & $10^{-1}$ & $7$ & $1$ & $1$ \\
  & $10^{-2}$ & $7$ & $2$ & $1$ \\
  & $10^{-3}$ & $8$ & $3$ & $2$ \\
  & $10^{-4}$ & $11$ & $6$ & $3$ \\
  & $10^{-5}$ & $13$ & $10$ & $5$ \\
  & $10^{-6}$ & $13$ & $13$ & $9$ \\
\addlinespace
\multirow{7}{*}{$100$}  & $10$ & $5$ & $1$ & $1$ \\
  & $1$ & $5$ & $1$ & $1$ \\
  & $10^{-1}$ & $5$ & $1$ & $1$ \\
  & $10^{-2}$ & $5$ & $1$ & $1$ \\
  & $10^{-3}$ & $4$ & $2$ & $1$ \\
  & $10^{-4}$ & $3$ & $3$ & $2$ \\
  & $10^{-5}$ & $3$ & $3$ & $3$ \\
  & $10^{-6}$ & $3$ & $3$ & $3$ \\
\addlinespace
\multirow{7}{*}{$200$}  & $10$ & $4$ & $1$ & $1$ \\
  & $1$ & $4$ & $1$ & $1$ \\
  & $10^{-1}$ & $4$ & $1$ & $1$ \\
  & $10^{-2}$ & $3$ & $1$ & $1$ \\
  & $10^{-3}$ & $1$ & $1$ & $1$ \\
  & $10^{-4}$ & $1$ & $1$ & $1$ \\
  & $10^{-5}$ & $1$ & $1$ & $1$ \\
  & $10^{-6}$ & $1$ & $1$ & $1$ \\
\bottomrule
\end{tabular}
\end{table}	
In most cases, when the computational cost has a negligible impact with respect to the communication cost, the optimal solution is to increase the level of aggregation.
This is quite interesting, and somewhat counter-intuitive, as increasing aggregation means spending more on moving data initially collected to fewer collection points, which requires using a ``costly" resource, i.e., communication. However, aggregating data then produces more precise local models at a small increase of cost, due to the low cost of computation, also reducing the number of rounds to achieve the target accuracy, thus eventually saving on the ``costly" resource. For example, note the case in the last row of the first block the optimal point reduces by a factor between 13 and 9 the number of collection points with respect to the set of nodes generating data. However, it is also interesting to note that neither in this case full aggregation (i.e., $\gamma=m_0$) is the best choice.

Conversely, as the computation and communication costs become similar, it is more efficient to move less the data, ``pay" less at collection points in terms of computation (remember that, due to our computation cost function model, computing on the same amount of data is more costly on a single node than separately on two nodes), even though this requires additional communication cost due to a higher number of rounds to achieve the required accuracy. 

An interesting behaviour is shown at the last three blocks of the sub-linear regime ($\hat\gamma_{\alpha=0.5}$) where within the block the optimal aggregation \emph{decreases} as the cost for computing becomes more and more negligible, which is the opposite of what we have just observed. This is actually a joint effect of the interplay between the two components of the cost, and the high numerosity of data available at each node ($n_0$). On the one hand, as observed before, when the computation cost becomes negligible compared to the communication cost, aggregating more becomes convenient, as the fewer collection points can work with a higher number of data, thus local models are more precise and lower communication rounds are needed to achieve a target accuracy. This is indeed what happens at low values of $n_0$. However, aggregating more means moving more data from original devices to the (fewer) collection points. This \emph{increases} the networking cost. At high values of $n_0$, i.e., when a large amount of data is already available at individual devices, this additional cost is not balanced by a reduced networking cost due to fewer rounds of communication required to converge. In other words, low levels of aggregation result in precise enough local models not to require a large number of communication rounds, which pushes towards lower aggregation levels. However, increased computation costs might change the picture. When computation becomes costly with respect to communication (high values of $\mu$) the overall computation cost becomes predominant over communication costs. Therefore, it pays of to aggregate more, as this allows for fewer rounds to refine the models, and thus to a lower total computation cost.

Summarising from the obtained results, we can draw a few general remarks. First, ``decentralisation helps". In all examined cases, the optimal operating point is far from centralising all data in a single collection point. This tells that decentralised machine learning, in addition to advantages in terms of privacy, is also beneficial in terms of efficiency, without compromising on model accuracy. Second, we have also shown that the \emph{shape} of the computation cost function with respect to the aggregation level has a key role in moving the optimal point between higher and lower aggregation. This might give interesting indications (and a useful tool) to  configure ``computational offers" by distributed infrastructure providers, which might tune, e.g., the cost of their VM resources also based on whether they prefer to work with higher or lower levels of aggregation. Third, we have shown that the relative cost of communication vs computation might also play a big role in setting the optimal operating point. This might lead to apparently counter-intuitive results: high communication costs might result in \emph{few} collection points aggregating more data (which have to be moved from the originating nodes), if this allows to obtain significantly more precise local models, or \emph{more} collection points if the amount of data generated at individual devices is large enough to obtain sufficiently accurate local models. High computation costs might result in \emph{few} collection points (where computation is more expensive) if the amount of data at individual nodes it too little to sustain precise enough local models.
\section{Conclusion}
\label{sec:conclusions}
In this paper we consider a scenario in which edge devices (IoT, personal mobile devices) accomplish a data analytics process on some generated or collected data. This is an increasingly relevant scenario in many use cases, as it is becoming more and more clear that distributed data analytics (machine learning, in our specific case) may bring significant advantages with respect to conventional centralised approaches. For example, in Industry 4.0 contexts this addresses data confidentiality and real-time operational constraints. However, in principle centralising data analytics is considered more effective, as machine learning models can be trained with more data. In this paper we addressed the research question whether this is actually the case, when resource consumption is also considered. Specifically, decentralised machine learning algorithms allow to collaboratively train models without sharing data, via successive rounds of communications through which local models are exchanged and fine tuned. Any target accuracy can be obtained, at the cost of increasing the number of refinement rounds. As each round has an associated cost in terms of communication traffic and computation, in this paper we investigate what is the optimal  level at which to aggregate raw data (i.e., the optimal number of nodes where to collect subsets of the raw data) such that the data analytics process achieves a target accuracy, while minimising the overall communication and computation resources.

To this end we propose an analytical framework able to cope with this problem. 
We exemplify the complete derivation of the optimal operating point in the case of one of the reference distributed learning algorithms, named DSVRG, which represents a large number of analytics tasks (e.g., classification and regression). Exploiting analytical expressions of the number of data exchanged between the nodes to run the decentralised algorithm, we derive a closed form expression for the computational and network resources required to achieve a target accuracy. 
We thus solve the model to obtain the optimal operating point, and study the properties of the cost function in the range of the possible levels of aggregation, ranging from complete decentralisation to complete centralisation. The analysis highlights that in most of the cases the optimal operating point is neither of these two extremes aggregation points. Rather, the typical optimal solution consists in aggregating raw data in an intermediate number of collection points, which run the decentralised machine learning algorithm in a collaborative fashion. The paper presents closed form expressions for the optimal operating point in significant cases, validates the accuracy of the analytical model, quantifies the advantage of configuring the system at the optimal operating point as opposed to the two extreme cases, and finally presents a sensitivity analysis showing the impact of various parameters on the optimal configuration. Moreover, we highlight that our analytical model not only provides a performance evaluation tool, but can also be used as a design too, to find (analytically or numerically) the optimal configuration of a decentralised machine learning system.


\section*{Acknowledgments}
This work is partially supported by two projects: HumanE AI Network (EU H2020 HumanAI-Net, GA \#952026), SoBigData++ (EU H2020 SoBigData++, GA \#871042) and Operational Knowledge from Insights and Analytics on Industrial Data (MIUR PON OK-INSAID, ARS01\_00917)
%
%
%
%
%
%
\bibliographystyle{IEEEtran}
\bibliography{./bibliography.bib}


\begin{IEEEbiography}[{\includegraphics[width=1in,height=1.25in,clip,keepaspectratio]{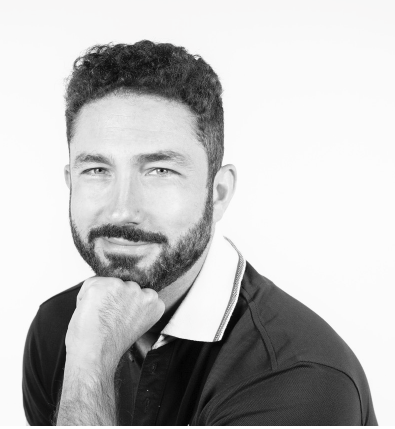}}]{Lorenzo Valerio}
is a Technologist at IIT-CNR. 
His current main research activity focuses on the design of communication efficient distributed learning solutions, machine learning solutions for resource-constrained devices and machine learning based cellular traffic offloading solutions for the Future Internet. 
He has published in journals and conference proceedings more than 30 papers. 
He has served as Workshop co-chair for IEEE AOC’15. He has been guest editor for the Elsevier Computer Communications and he is currently member of the editorial board for the same Journal. He  He has been recipient for 1 Best Paper Award at IEEE WoWMoM 2013 and one Best Paper Nomination at IEEE SMARTCOMP 2016.
\end{IEEEbiography}

\begin{IEEEbiography}[{\includegraphics[width=1in,height=1.25in,clip,keepaspectratio]{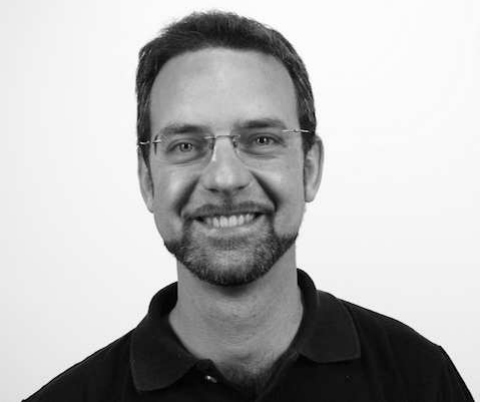}}]{Andrea Passarella}
is a Research Director at IIT-CNR and Head of the Ubiquitous Internet Group.
Before joining UI-IIT he was a Research Associate at the Computer
Laboratory of the University of Cambridge, UK.
He published 150+ papers in international journals and conferences,
receiving 4 best paper awards, including at IFIP Networking 2011
and IEEE WoWMoM 2013.
He was Chair/Co-Chair of several IEEE and ACM conferences/workshops
(including IFIP IoP 2016, ACM CHANTS 2014 and IEEE WoWMoM 2011 and 2019).
He is the founding Associate EiC of the Elsevier journal Online
Social Networks and Media (OSNEM), and Area Editor for the Elsevier
Pervasive and Mobile Computing Journal and Inderscience Intl.
Journal of Autonomous and Adaptive Communications Systems.
\end{IEEEbiography}

\begin{IEEEbiography}[{\includegraphics[width=1in,height=1.25in,clip,keepaspectratio]{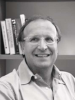}}]{Marco Conti}
is the Director of IIT-CNR.
He was the coordinator of FET-open project ``Mobile Metropolitan
Ad hoc Network (MobileMAN)'' (2002-2005), and he has been the CNR
Principal Investigator (PI) in several projects funded by the
European Commission: FP6 FET HAGGLE (2006-2009), FP6 NEST MEMORY
(2007-2010), FP7 EU-MESH (2008-2010), FP7 FET SOCIALNETS (2008-2011),
FP7 FIRE project SCAMPI (2010-2013), FP7 FIRE EINS (2011-15) and
CNR Co-PI for the FP7 FET project RECOGNITION (2010-2013).
%
%
He has published in journals and conference proceedings more than
300 research papers
to design, modelling, and performance evaluation of computer and
communications systems.
%
\end{IEEEbiography}

\onecolumn
\appendices

\section{Proof of Theorem \ref{th:convex}}
\label{app:proofthcvx}
Here we prove Theorem~\ref{th:convex}. To this end we analyse  Eq.~\ref{eq:cg} (reported here for convenience) and we define the conditions for which each of its terms is convex. 
\begin{eqnarray}
	C(\gamma)  &=& C_A(\gamma) + C_D(\gamma)+C_P(\gamma) \nonumber \\
    &= &  2 w \theta \left(\frac{m_0}{\gamma}-1\right)R(\gamma) +\nonumber\\
    & & +   \theta\left(m_0 - \frac{m_0}{\gamma}\right)n_0 d +\nonumber\\
    & & +   \beta f(\gamma) \tau(n_0m_0+n_0\gamma)R(\gamma)
\end{eqnarray}
First let us set the domain assumptions of the following constants: $w,m_0,n_0,d,\tau\in[1,\infty] \subseteq \mathbb{R}$, $\theta,\beta \in (0,\infty)\subseteq \mathbb{R}$ , $\varepsilon \in (0,1) \subseteq \mathbb{R}$. Moreover we recall that $\gamma \in [1,\infty) \subseteq \mathbb{R}, \alpha \in(0,\infty) \subseteq \mathbb{R}$ and we assume $k \leq \sqrt{N}$ as also discussed in \cite{zhang2013linear,Lee:2015aa,Lee:2017aa}.


\subsection*{Convexity analysis of $C_A(\gamma)$}
Let us consider the second derivative of $C_A(\gamma)$: 
\begin{equation}
    C_A''(\gamma) = \frac{4 \theta  \lgg w (-\gamma  \kappa+3 \kappa m_0+\gamma  m_0 n_0)}{\gamma ^4 n_0}
    \label{d2ca}
\end{equation}
According to the domain assumptions, all the terms of Eq. (\ref{d2ca}) are positive. Therefore, we just need to identify the condition on $\gamma$ for which holds the inequality
\begin{eqnarray}
-\gamma k + 3km_0 + \gamma m_0n_0 &\geq& 0 \nonumber\\
\gamma (m_0n_0 - k) &\geq& -3km_0 \nonumber \\
\gamma &\geq& \frac{-3km_0}{m_0n_0 - k}
\end{eqnarray}
for $k\neq m_0n_0$. Provided that we assume $\kappa \leq \sqrt{N} < m_0n_0$, the condition is verified. 
Therefore, since in the domain of interest the interval $[1,\infty) \subset [\frac{-3\kappa m_0}{m_0n_0-\kappa},\infty)$ we conclude that for $\gamma\in[1,\infty)$, $C_A(\gamma)$ is convex.
%
\subsection*{Convexity analysis of $C_D(\gamma)$}   
Let us consider the second derivative of $C_D(\gamma)$:
\begin{equation}
    C''_D(\gamma) = \frac{- 2 d m_0 n_0 \theta}{\gamma^3} 
    \label{d2cd}
\end{equation}
Considering that according to the domain assumptions all if its terms are positive, the second derivative $C_D''(\gamma)$ is always negative. Therefore, the term $C_D(\gamma)$ is \emph{concave}.
%
\subsection*{Convexity analysis of $C_P(\gamma)$}
Let us consider the second derivative:
\begin{equation}
    C''_P(\gamma) = \beta  \lgg \tau  \gamma ^{\alpha -3} \left(\alpha ^2 \gamma  \kappa-\alpha  \gamma  \kappa+\alpha ^2 \kappa m_0-3 \alpha  \kappa m_0+2 \kappa m_0+\alpha ^2 \gamma  m_0 n_0-\alpha  \gamma  m_0 n_0+\alpha ^2 \gamma ^2 n_0+\alpha  \gamma ^2 n_0\right)
    \label{d2cp}
\end{equation}
We consider the following limits: 
\begin{eqnarray}
    \lim_{\gamma \rightarrow \infty} C_P''(\gamma) &=& +\infty  \label{eq:cp_lim_inf}\\
    \lim_{\gamma \rightarrow 1^+} C_P''(\gamma) &=& \beta  \lgg \tau  \left(\alpha ^2 \kappa-\alpha  \kappa+\alpha ^2 \kappa m_0-3 \alpha  \kappa m_0+2 \kappa m_0+\alpha ^2 m_0 n_0-\alpha  m_0 n_0+\alpha ^2 n_0+\alpha  n_0\right)\label{eq:cp_lim_one}
\end{eqnarray}
Provided that, according to Eq. (\ref{eq:cp_lim_inf}) $C''_P(\gamma)$ is always positive for increasing values of $\gamma$, we concentrate on Eq. (\ref{eq:cp_lim_one}) where we study when it is positive. 
To this end, we analyse the following Equation (\ref{eq:cp_focus}) with respect to $\alpha$. 
\begin{equation}
P(\alpha) = \beta  \lgg \tau  (\alpha ^2 \kappa-\alpha  \kappa+\alpha ^2 \kappa m_0-3 \alpha  \kappa m_0+2 \kappa m_0+\alpha ^2 m_0 n_0-\alpha  m_0 n_0+\alpha ^2 n_0+\alpha  n_0)\geq0 \label{eq:cp_focus}
\end{equation}
Note that Eq. \ref{eq:cp_focus} is positive when $\alpha \leq \alpha_1$ and $\alpha \leq \alpha_2$, where $\alpha_1,\alpha_2$ are the two roots of Eq. \ref{eq:cp_focus} with the following expression: 
\begin{eqnarray}
\alpha_1 &\leq& \frac{+3 \kappa m_0+\kappa+m_0 n_0-n_0-\sqrt{(-3 \kappa m_0-\kappa-m_0 n_0+n_0)^2-8 \kappa m_0 (\kappa m_0+\kappa+m_0 n_0+n_0)}}{2 (\kappa m_0+\kappa+m_0 n_0+n_0)} \\
\alpha_2 &\geq& \frac{+3 \kappa m_0+\kappa+m_0 n_0-n_0+\sqrt{(-3 \kappa m_0-\kappa-m_0 n_0+n_0)^2-8 \kappa m_0 (\kappa m_0+\kappa+m_0 n_0+n_0)}}{2 (\kappa m_0+\kappa+m_0 n_0+n_0)}  
\end{eqnarray}
Now we are interested in understanding where $\alpha_1,\alpha_2$ are positioned on the Real line, i.e., with respect to the value $\alpha=1$. To this end, we consider the following limits. 
\begin{eqnarray}
\lim_{\alpha \rightarrow +\infty} P(\alpha) &=& +\infty \label{eq:lia1}\\
\lim_{\alpha \rightarrow 1^+} P(\alpha) &=& 2n_0 > 0\label{eq:lia2} \\
\lim_{\alpha \rightarrow 0^+} P(\alpha) &=& 2\kappa m_0 > 0\label{eq:lia3} \\
\lim_{\alpha \rightarrow -\infty} P(\alpha) &=& +\infty\label{eq:lia4} 
\end{eqnarray}
 which suggest that $0<\alpha_1<\alpha_2<1$. To confirm it we calculate the value $\alpha_0$ for which $P'(\alpha) = 0$.  
\begin{eqnarray}
P'(\alpha) = \beta  \lgg \tau  (2 \alpha  \kappa+2 \alpha  \kappa m_0-3 \kappa m_0-\kappa+2 \alpha  m_0 n_0-m_0 n_0+2 \alpha  n_0+n_0)
\end{eqnarray}
 The first derivative of $P(\alpha)$ is null for 
\begin{equation}
   \alpha_0 =  \frac{3 \kappa m_0+\kappa+m_0 n_0-n_0}{2 (m_0+1) (\kappa+n_0)}
\end{equation}
We observe that i) $\alpha_0 >0\ \forall k,n,m >1$ (which is compliant with the initial assumptions) and ii) $\alpha_0 < 1$ for $k < \frac{(m+3) n}{m-1}$. 

Therefore, we conclude that when $\alpha\leq\alpha_1 \vee\alpha\geq \alpha_2$, $C''_P(\gamma)\geq 0$ which implise that $C_P(\gamma)$ is convex. With respect to our model, this includes the linear, the super-linear and some of the sub-linear cases.

Finally, since the term $C_D(\gamma)$ is concave, we can conclude that: (i) $C_A(\gamma)+C_P(\gamma)$  is convex if $\alpha\leq\alpha_1 \vee\alpha\geq \alpha_2$ and (ii) when $C_D(\gamma)$ is negligible compared to $C_A(\gamma)+C_P(\gamma)$, Eq. (\ref{eq:cg}) is convex. This concludes the proof.\qed

\section{Proof of Theorem \ref{th:lin}}
\label{app:profthlin}
We want to prove that the Equation~\ref{eq:cg} (reported here for convenience) has only one minimum in the range $[1,m_0]$.
To this end we prove that, i) there is only one solution in $[1,m_0]$ and ii) that it is a minimum for the function $C(\gamma)$. Note we conducted this analysis by means of Sympy (http://www.sympy.org), a standard and open-source mathematical manipulation software. Therefore, our analysis is completely reproducible.  

Considering the equation
\begin{eqnarray}
	C(\gamma)  &=& C_A(\gamma) + C_D(\gamma)+C_P(\gamma) \nonumber \\
    &= &  2 w \theta \left(\frac{m_0}{\gamma}-1\right)R(\gamma) +\nonumber\\
    & & +   \theta\left(m_0 - \frac{m_0}{\gamma}\right)n_0 d +\nonumber\\
    & & +   \beta f(\gamma) \tau(n_0m_0+n_0\gamma)R(\gamma)
\end{eqnarray}

and its first derivative $C'(\gamma)$ in $\gamma$:

\begin{multline}
 C'(\gamma) = \\ \frac{1}{\gamma^{3} n_0}\Bigg(\beta \gamma^{3} \lgg n_0^{2} \tau \left(\gamma + m_0\right) + \beta \gamma^{3} \lgg n_0 \tau \left(\gamma n_0 + k\right) + d \gamma m_0 n_0^{2} \theta + 2 \lgg \omega \theta \left(k \left(\gamma - m_0\right) - m_0 \left(\gamma n_0 + k\right)\right)\Bigg)
\end{multline}

it is possible to find, after some algebraic manipulations, the four roots  of Eq. $C'(\gamma)=0$ w.r.t. $\gamma$, all in closed-form and denoted as $\gamma_1,\gamma_2,\gamma_3,\gamma_4$. Moreover,  $\gamma_1 \in \mathbb{R}^+$, $\gamma_2 \in \mathbb{R}^-$ and $\gamma_3,\gamma_4 \in \mathbb{C}$. Since in our problem $\gamma \in [1,m_0] \subseteq \mathbb{R}^+$, the only solution of interest is $\gamma_1$. The definition of $\gamma_1$ provided in Appendix \ref{app:defs}, while the rest of the roots are not reported for the sake of space and clarity, i.e., their form is quite complex to render (see Appendix \ref{app:defs}) therefore we provide only the one useful for the paper.

We conclude the first part of the proof stating that there exists one and only one solution in the domain of our problem. The proof that $\gamma_1$ is a minimum for $C(\gamma)$ follows directly from the proof Theorem \ref{th:convex} presented in Appendix \ref{app:profthlin}. \qed



\section{Analytical formula of $\gamma_1$ }
\label{app:defs}
The function $\gamma_1$ of Theorem~\ref{th:lin} is defined as:
\begin{equation}
\gamma_1 = A(m,n,d,k,\tau,\mu,\varepsilon,\theta,\beta,\omega)
\end{equation}
where $A(m,n,d,k,\tau,\mu,\varepsilon,\theta,\beta,\omega)$, is a function  defined as follows: 

\begin{landscape}
\tiny
\begin{multline}
  -\frac{1}{24 n_{0\omega}}\Bigg( n_0m_0(3 k+3 m_0 n_0+\sqrt{3} n_0 \Bigg(\frac{1}{n_0^2}[3 (k+m_0 n_0)^2+\frac{1}{\beta  \lgg \tau }\Bigg[8 \sqrt[3]{3} (-18 k
    \lgg m_0 \theta  \omega \Big(\beta  \lgg m_0 \tau  n_0^2+\beta  k \lgg \tau  n_0\Big)^2+9 \beta  \lgg n_0^2 \tau
     \theta ^2 \Big(d m_0 n_0^2+2 \lgg (k-m_0 n_0) \omega\Big)^2+\sqrt{3} \\\Big(\lgg^2 \theta ^2 \Big(\beta ^3 \lgg
    n_0^3 \tau ^3 \theta  \Big(d m_0 (k+m_0 n_0) n_0^2+2 \lgg \Big(k^2+16 m_0 n_0 k-m_0^2 n_0^2\Big) \omega\Big)^3+ 27 \Big(\beta  \tau
    \theta  \Big(d m_0 n_0^3+2 \lgg (k-m_0 n_0) \omega n_0\Big)^2-2 k m_0 \Big(\beta  \lgg m_0 \tau  n_0^2+\beta  k \lgg \tau
    n_0\Big)^2 \omega\Big)^2\Big)\Big)^{1/2})^{1/3}\Bigg]\\-\Big[8\ 3^{2/3} n_0 \theta  \Big(d m_0 (k+m_0 n_0) n_0^2+2 \lgg
    \Big(k^2+16 m_0 n_0 k-m_0^2 n_0^2\Big) \omega\Big)\Big]\Big[\Big(-18 k \lgg m_0 \theta  \omega \Big(\beta  \lgg m_0 \tau
    n_0^2+\beta  k \lgg \tau  n_0\Big)^2+9 \beta  \lgg n_0^2 \tau  \theta ^2 \Big(d m_0 n_0^2+2 \lgg (k-m_0 n_0)
    \omega\Big)^2+\sqrt{3} \\\sqrt{\lgg^2 \theta ^2 \Big(\beta ^3 \lgg n_0^3 \tau ^3 \theta  \Big(d m_0 (k+m_0 n_0) n_0^2+2
    \lgg \Big(k^2+16 m_0 n_0 k-m_0^2 n_0^2\Big) \omega\Big)^3+27 \Big(\beta  \tau  \theta  \Big(d m_0 n_0^3+2 \lgg (k-m_0 n_0) \omega
    n_0\Big)^2-2 k m_0 \Big(\beta  \lgg m_0 \tau  n_0^2+\beta  k \lgg \tau  n_0\Big)^2
    \omega\Big)^2\Big)}\Big)^{1/3}\Big]^{-1}]\Bigg)^{1/2} \\ -\sqrt{6} n_0 \Bigg\{\frac{1}{n_0^3}\Bigg[\Big(4\ 3^{2/3} \theta  \Big(d m_0 (k+m_0 n_0) n_0^2+2 \lgg
    \Big(k^2+16 m_0 n_0 k-m_0^2 n_0^2\Big) \omega\Big) n_0^2\Big) \Big(\Big(-18 k \lgg m_0 \theta  \omega \Big(\beta  \lgg m_0 \tau
    n_0^2+\beta  k \lgg \tau  n_0\Big)^2+9 \beta  \lgg n_0^2 \tau  \theta ^2 \Big(d m_0 n_0^2+2 \lgg (k-m_0 n_0)
    \omega\Big)^2+\sqrt{3} \\\sqrt{\lgg^2 \theta ^2 \Big(\beta ^3 \lgg n_0^3 \tau ^3 \theta  \Big(d m_0 (k+m_0 n_0) n_0^2+2
    \lgg \Big(k^2+16 m_0 n_0 k-m_0^2 n_0^2\Big) \omega\Big)^3+27 \Big(\beta  \tau  \theta  \Big(d m_0 n_0^3+2 \lgg (k-m_0 n_0) \omega
    n_0\Big)^2-2 k m_0 \Big(\beta  \lgg m_0 \tau  n_0^2+\beta  k \lgg \tau  n_0\Big)^2 \omega\Big)^2\Big)}\Big)^{1/3}\Big)^{-1}+ \\3 (k+m_0 n_0)^2
    n_0-\frac{1}{\beta  \lgg \tau
    }\Big[4 \sqrt[3]{3} \Big(-18 k \lgg m_0 \theta  \omega \Big(\beta  \lgg m_0 \tau  n_0^2+\beta  k \lgg \tau
    n_0\Big)^2+9 \beta  \lgg n_0^2 \tau  \theta ^2 \Big(d m_0 n_0^2+2 \lgg (k-m_0 n_0) \omega\Big)^2+\sqrt{3} \\
    \sqrt{\lgg^2 \theta ^2 \Big(\beta ^3 \lgg n_0^3 \tau ^3 \theta  \Big(d m_0 (k+m_0 n_0) n_0^2+2 \lgg \Big(k^2+16
    m_0 n_0 k-m_0^2 n_0^2\Big) \omega\Big)^3+27 \Big(\beta  \tau  \theta  \Big(d m_0 n_0^3+2 \lgg (k-m_0 n_0) \omega n_0\Big)^2-2 k m_0
    \Big(\beta  \lgg m_0 \tau  n_0^2+\beta  k \lgg \tau  n_0\Big)^2 \omega\Big)^2\Big)})^{1/3} n_0\Big]+\\
    \Big(3 \sqrt{3} \Big(\beta  \lgg \tau  (k+m_0 n_0)^3+32 n_0 \theta  \Big(d m_0 n_0^2+2 \lgg (k-m_0 n_0)
    \omega\Big)\Big)\Big)\Big(\beta  \lgg \tau   \Big(\frac{1}{n_0^2}\Bigg[3 (k+m_0 n_0)^2+\frac{1}{\beta  \lgg \tau }\\(8 \sqrt[3]{3} (-18 k \lgg m_0 \theta  \omega
    \Big(\beta  \lgg m_0 \tau  n_0^2+\beta  k \lgg \tau  n_0\Big)^2+9 \beta  \lgg n_0^2 \tau  \theta ^2 \Big(d m_0
    n_0^2+2 \lgg (k-m_0 n_0) \omega\Big)^2+\sqrt{3} \\\sqrt{\lgg^2 \theta ^2 \Big(\beta ^3 \lgg n_0^3 \tau ^3 \theta
    \Big(d m_0 (k+m_0 n_0) n_0^2+2 \lgg \Big(k^2+16 m_0 n_0 k-m_0^2 n_0^2\Big) \omega\Big)^3+27 \Big(\beta  \tau  \theta  \Big(d m_0
    n_0^3+2 \lgg (k-m_0 n_0) \omega n_0\Big)^2-2 k m_0 \Big(\beta  \lgg m_0 \tau  n_0^2+\beta  k \lgg \tau  n_0\Big)^2
    \omega\Big)^2\Big)})^{1/3})\\- (8\ 3^{2/3} n_0 \theta  \Big(d m_0 (k+m_0 n_0) n_0^2+2 \lgg \Big(k^2+16 m_0
    n_0 k-m_0^2 n_0^2\Big) \omega\Big))((-18 k \lgg m_0 \theta  \omega \Big(\beta  \lgg m_0 \tau  n_0^2+\beta  k \lgg
    \tau  n_0\Big)^2+9 \beta  \lgg n_0^2 \tau  \theta ^2 \Big(d m_0 n_0^2+2 \lgg (k-m_0 n_0) \omega\Big)^2+\sqrt{3}
    \\ (\lgg^2 \theta ^2 \Big(\beta ^3 \lgg n_0^3 \tau ^3 \theta  \Big(d m_0 (k+m_0 n_0) n_0^2+2 \lgg \Big(k^2+16
    m_0 n_0 k-m_0^2 n_0^2\Big) \omega\Big)^3+\\27 \Big(\beta  \tau  \theta  \Big(d m_0 n_0^3+2 \lgg (k-m_0 n_0) \omega n_0\Big)^2-2 k m_0
    \Big(\beta  \lgg m_0 \tau  n_0^2+\beta  k \lgg \tau  n_0\Big)^2 \omega\Big)^2\Big))^{1/2})^{1/3})^{-1}\Bigg]\Big)^{1/2}\Bigg\}^{-1}\Bigg]\Bigg)^{1/2}\Bigg)
\end{multline}
\end{landscape}
\end{document}